\documentclass[aip,jcp,preprint,amsmath,amssymb,floatfix,nofootinbib]{revtex4-2}

\usepackage[utf8]{inputenc}
\usepackage[T1]{fontenc}
\usepackage{graphicx}
\graphicspath{{./}{./figures/}{./figures/si/}{./si_figures/}}
\usepackage{amsmath,amssymb,amsthm}
\usepackage{bm}
\usepackage{booktabs}
\usepackage{longtable}
\usepackage{threeparttable}
\usepackage{array}
\usepackage{enumitem}
\usepackage[hidelinks]{hyperref}
\usepackage{xcolor}
\usepackage{microtype}
\usepackage{needspace}
\usepackage{tikz}
\usetikzlibrary{arrows.meta,positioning,fit,shapes.geometric,calc}

\let\oldsection\section
\renewcommand{\section}{\Needspace{8\baselineskip}\oldsection}

\setlist[itemize]{topsep=3pt,itemsep=2pt,parsep=0pt}
\setlist[enumerate]{topsep=3pt,itemsep=2pt,parsep=0pt}

\newtheorem{theorem}{Theorem}
\newtheorem{proposition}[theorem]{Proposition}
\theoremstyle{definition}
\newtheorem{definition}[theorem]{Definition}

\newcommand{\Cset}{\mathcal{C}}
\newcommand{\Rbb}{\mathbb{R}}

\newcommand{\Nrm}{\mathcal{N}}
\newcommand{\Tmap}{\mathcal{T}}
\newcommand{\one}{\mathbf{1}}

\newcommand{\Trans}{\mathsf{T}}
\newcommand{\Fish}{\mathcal{I}}
\newcommand{\wvec}{\bm{w}}
\newcommand{\yvec}{\bm{y}}
\newcommand{\muvec}{\bm{\mu}}
\newcommand{\Sig}{\Sigma}
\DeclareMathOperator{\rank}{rank}

\DeclareMathOperator*{\argmax}{arg\,max}
\DeclareMathOperator*{\argmin}{arg\,min}

\newcommand{\Exp}{\mathbb{E}}
\newcommand{\wavenum}{\ensuremath{\,\mathrm{cm}^{-1}}}
\newcommand{\angstrom}{\text{\AA}}

\newcommand{\DepositDOI}{10.5281/zenodo.21712522}
\newcommand{\DepositURL}{https://doi.org/\DepositDOI}
\newcommand{\DepositLink}{\expandafter\url\expandafter{\DepositURL}}

\begin{document}

\title{Finite-Resolution Identifiability and Measurement Design for
Molecular Conformer Spectroscopy}

\author{Megan Simons}
\email{msimons@recognitionphysics.org}
\affiliation{Recognition Physics Institute, Austin, Texas, USA}
\author{Jonathan Washburn}
\affiliation{Recognition Physics Institute, Austin, Texas, USA}

\begin{abstract}
Conformer assignment is meaningful only when the measurement can distinguish
the candidate structures. We formulate conformer spectroscopy as a
measurement-model-dependent identifiability problem in which each modality
defines an observation law with explicit experimental and theoretical
uncertainty. Equality of observation laws defines exact classes that refine as
modalities are added, whereas finite-resolution ambiguity is non-transitive and
is represented by a Bayes-error graph. We also give rank and
constrained-Fisher criteria for population recovery from unnormalized additive
spectra.
The framework is applied to audited B3LYP-D3(BJ)/def2-TZVP ensembles of
1,2-difluoroethane, ethylene glycol, and \emph{n}-pentane. Under the declared
shared-covariance working model, IR separates all non-mirror pairs, leaving
only mirror pairs that are exactly degenerate under the achiral observation
maps. This within-method result does not imply general IR sufficiency: in a
six-case PBE0-to-B3LYP diagnostic, fixed calibration recovers one intended
representative, whereas scale-and-shift profiling with a conservative
cross-method covariance recovers five representatives and all six achiral
classes. Under a combined stress-test covariance, \emph{n}-pentane develops one
non-mirror quotient ambiguity. Three individual Raman windows remove it, while
no tested IR window does. Mirror-collapsed class populations remain identifiable
and well conditioned under the calibrated-scale working model, whereas separate
mirror-partner populations are exactly unidentifiable from achiral additive
spectra.
\end{abstract}


\maketitle

\section{Introduction}
\label{sec:intro}

\subsection{Distinguishability as a prerequisite for conformer assignment}

Assigning molecular conformers from vibrational spectra is a long-standing
experimental and computational problem. Temperature-dependent infrared and Raman
intensities, matrix isolation, and related vibrational methods have long been
used to identify rotamers and estimate conformational energy
differences~\cite{Klaeboe1995,Durig1992DFE}. Microwave spectroscopy likewise
assigns conformers from rotational constants and dipole
components~\cite{GordyCook1984}, while additive spectral mixtures have been
treated by self-modeling curve resolution and related chemometric
unmixing~\cite{LawtonSylvestre1971}. More recently, the IR spectra alignment
(IRSA) algorithm combines infrared spectra with Raman or vibrational circular
dichroism (VCD) and uses pseudo-Voigt deconvolution and frequency scaling to
assign relative stereochemistry~\cite{IRSA2023}. A deep-learning method,
Vib2Conf (arXiv preprint), discriminates three-dimensional conformers from
vibrational spectra and reports top-1 recall exceeding $95\%$ on conventional
spectrum-to-structure benchmarks and $82.06\%$ on a near-isomeric conformer test
set with root-mean-square deviations (RMSD) near $1\,\angstrom$~\cite{Vib2Conf2026}.
These approaches primarily ask whether a model can predict a conformer from a
spectrum. That question is important, but it is downstream of a more basic
one. A
classifier can be accurate on a benchmark while providing no account of
why it succeeds, no model-conditional certificate that a given pair of
conformers is even separable under the declared observation law and covariance,
and no guidance about what to measure next when it fails. We therefore pose the
prior question: given a specified experimental resolution, noise level,
theoretical uncertainty, and set of available spectroscopic modalities, which
conformers are actually distinguishable, and what additional observation would
resolve the remaining ambiguity?

We instead formulate the problem as one of identifiability and experimental
design: we ask what information the physics of the measurement makes
available, independently of any particular estimator, and we ask how to
allocate measurement effort optimally. The distinction matters because the answer to the design question
is an invariant of the measurement model: a model-conditional pairwise Bayes
performance limit at equal priors, rather than a general multiclass top-1 bound
for arbitrary priors. Classifier accuracy, by contrast, is a property of a
particular trained model and its benchmark.

\subsection{Contributions and scope}
\label{sec:contributions}

The proposed formulation complements classifier accuracy with
model-conditional measures of pairwise ambiguity, population identifiability,
and measurement value: an ambiguity certificate at a stated resolution and
error level; a pairwise decision-error estimate rooted in statistical decision
theory; an explanation of which added modality or spectral window resolves a
given ambiguity; criteria for whether conformer-population recovery is
identifiable for additive spectra before inversion is attempted; and
sensitivity of those conclusions to the working covariance and
conformer-ensemble curation. The specific contributions are a clean separation
between exact observational equivalence, an equivalence relation that induces a
quotient and refines monotonically as modalities are added
(Theorem~\ref{thm:refinement}), and finite-noise ambiguity, an explicitly
non-transitive graph (Proposition~\ref{prop:nontransitive}), together with the
model-conditional Bayes-error measure that builds it
(Section~\ref{sec:finiteresolution}); rank and constrained-Fisher criteria for
exact and approximate conformer-population identifiability in additive mixtures
(Theorem~\ref{thm:identifiability} and
Eqs.~\eqref{eq:constrainedfisher}--\eqref{eq:fisherpsd}), connected below to a
numerical rank demonstration on achiral IR/Raman population columns;
experimental-design objectives for selecting the next modality or spectral
window for discrimination, with corresponding information criteria for
population recovery from additive spectra (Section~\ref{sec:design}); an
evaluation design that incorporates theoretical model discrepancy explicitly
rather than matching ideal spectra to themselves
(Sections~\ref{sec:protocol}--\ref{sec:pilot}); and an open-source
implementation, exercised on three molecules that span a difficulty gradient,
that demonstrated refinement and edge monotonicity, resolved a stressed
ambiguity down to individual spectral windows, and audited the dependence of
those results on covariance and conformer curation (Section~\ref{sec:pilot}).

Although the underlying mathematical ingredients are established, their
integration yields an operational framework linking conformer
distinguishability, population recovery, and measurement design, demonstrated
here on computed pure-conformer examples. Reflection invariance gives an exact
physical reason that mirror-related conformers coincide under the declared
achiral maps. Numerical edge counts and design implications below are
interpreted under the assumptions collected in Section~\ref{sec:discussion}.

\section{Observable maps and the measurement model}
\label{sec:observables}

\subsection{Conformers, modalities, and forward maps}

Fix a molecule $m$ with a finite candidate set of low-energy conformers
\begin{equation}
\Cset_m=\{c_1,c_2,\ldots,c_{n_m}\}.
\end{equation}
A spectroscopic modality $a$ is a physical observation channel such as
infrared absorption (IR), Raman scattering, or gas-phase rotational
spectroscopy. Each modality is associated with a forward observable map
\begin{equation}
S_a:\Cset_m\longrightarrow\Rbb^{d_a},
\qquad c_i\longmapsto S_a(c_i),
\end{equation}
where $S_a(c_i)$ is a discretized theoretical signal computed from the
electronic-structure description of $c_i$. Representative maps are
\begin{equation}
S_{\mathrm{IR}}(c_i),\qquad
S_{\mathrm{Raman}}(c_i),\qquad
S_{\mathrm{rot}}(c_i)=(A_i,B_i,C_i),\qquad
S_{\bm{\mu}}(c_i)=f_{\bm{\mu}}(c_i),
\end{equation}
where $S_{\mathrm{IR}}$ and $S_{\mathrm{Raman}}$ are binned intensity vectors
over a wavenumber grid and $S_{\mathrm{rot}}$ collects the three rotational
constants in the convention $A\ge B\ge C$ (equivalently, the ordered principal
moments of inertia $I_a\le I_b\le I_c$).

\paragraph{Scope of the conformer-state representation.}
A vertex $c_i$ denotes an optimized local minimum and $S_a(c_i)$ its harmonic
or rigid-rotor observable. This representation is appropriate when the minima
are localized on the measurement timescale or are used deliberately as a
structural library. Rapid interconversion, tunneling, and large-amplitude motion
instead require dynamical or spectroscopic states, such as vibrationally averaged
families or tunneling eigenstates. The ethylene-glycol calculations below should
therefore be interpreted as distinguishability calculations on a static
local-minimum library rather than as a model of its complete gas-phase
rotation--torsion eigenstate
structure~\cite{ChristenMuller2003,MullerChristen2004}. A \emph{declared mirror
pair} is a pair of minima identified as reflection images by the computational
audit.

The dipole channel requires a precise definition because a raw permanent-dipole
vector is neither rotation- nor reflection-invariant. We use the
reflection-invariant feature
\begin{equation}
f_{\bm{\mu}}(c_i)=\bigl(\,|\mu_{a,i}|,\ |\mu_{b,i}|,\ |\mu_{c,i}|\,\bigr),
\label{eq:dipolefeature}
\end{equation}
the magnitudes of the projections of the permanent dipole onto the principal
inertial axes $(a,b,c)$ ordered as above. The absolute principal-axis components
remove the sign ambiguity associated with axis orientation and reflection. Near
a principal-axis degeneracy, the individual components must be replaced by the
dipole norm within the degenerate subspace; the numerical guard used here is
documented in the supplementary material. Principal-axis dipole components are
obtained from Stark shifts or intensities of assigned rotational
transitions~\cite{GordyCook1984}, so $f_{\bm\mu}$ is treated as a feature of a
rotational-spectroscopy experiment rather than as an independently available
standalone modality.

These four observation channels were exercised in the study of
Section~\ref{sec:pilot}; VCD, UV--visible absorption, and X-ray absorption
(XAS) spectra fit the same observable-map structure and are natural extensions.

\subsection{From ideal spectra to measurements}

Ideal theoretical spectra are not measurements, and a framework that compared
only ideal spectra would establish nothing beyond the fact that spectra
computed by one program differ from one another. We therefore model the actual
observation from modality $a$ of conformer $c_i$ as
\begin{equation}
Y_a=\Tmap_a\!\left[S_a(c_i);\eta_a\right]+\varepsilon_a,
\label{eq:obsmodel}
\end{equation}
where
\begin{itemize}
\item $\Tmap_a$ is the measurement operator, applying instrumental
broadening (a line-shape convolution set by the resolution $\Delta\nu$),
frequency scaling, spectral shifts, baseline distortion, and normalization;
\item $\eta_a$ collects the nuisance parameters of $\Tmap_a$ (scale
factor, global and local shifts, baseline coefficients, broadening widths);
\item $\varepsilon_a$ is the error, comprising both experimental noise
and theoretical model discrepancy.
\end{itemize}

A convenient and tractable model takes the processed signal to be Gaussian
conditional on the nuisance parameters,
\begin{equation}
Y_a\mid C=c_i,\eta_a\;\sim\;\Nrm\!\left(\muvec_{a,i}(\eta_a),\,\Sig_a\right),
\qquad \muvec_{a,i}(\eta_a):=\Tmap_a\!\left[S_a(c_i);\eta_a\right],
\label{eq:gaussmodel}
\end{equation}
in which the covariance $\Sig_a$ carries both the measurement noise and
an estimate of the theoretical model-error covariance (Section~\ref{sec:noise}
specifies the working model). This is the essential modeling
decision that prevents the analysis from overstating distinguishability:
distinguishability is assessed against a covariance that includes theoretical
uncertainty, not against noiseless self-consistent spectra.

We distinguish three levels of the observation model, and are explicit about
which is used where.
\begin{enumerate}
\item the conditional Gaussian model~\eqref{eq:gaussmodel}, valid for a
fixed value of the nuisance parameters $\eta_a$;
\item the exact marginalized observation-law map, in which the nuisance parameters
are integrated against a prior $p(\eta_a)$,
\begin{equation}
p_a(y\mid c_i)=\int p_a(y\mid c_i,\eta_a)\,p(\eta_a)\,\mathrm{d}\eta_a;
\label{eq:marginal}
\end{equation}
this marginal is in general a non-Gaussian mixture, and its covariance
can depend on the conformer;
\item a declared Gaussian working approximation of~\eqref{eq:marginal},
obtained either by fixing $\eta_a$ at a profiled (best-alignment) value in the
spirit of IRSA~\cite{IRSA2023} or by moment matching, and used only for the
tractable pairwise separations of Section~\ref{sec:finiteresolution}.
\end{enumerate}
Marginalization~\eqref{eq:marginal} is the principled default because it
accounts for, rather than freely exploiting, the alignment freedom that any
real assignment must also expend. All closed-form Bayes-error statements in this
paper are computed under level (3) and are therefore conditional on that
approximation; where the marginal covariance is strongly conformer-dependent we
instead evaluate the Bayes error numerically or bound it
(Section~\ref{sec:finiteresolution}).

A note on normalization. The measurement operator $\Tmap_a$ includes an optional
intensity normalization (area or peak), which is a nonlinear operation:
$N(\sum_i w_i S_i)\ne\sum_i w_i N(S_i)$ for a mixture with populations $w_i$.
Normalization is harmless for pure-conformer discrimination
(Sections~\ref{sec:observationlaws}--\ref{sec:ambiguitygraph}), where each conformer
is processed on its own, but it breaks the additive mixture model of
Section~\ref{sec:populations}. There we therefore work with unnormalized
additive signals. The calibrated-scale and unknown-scale cases are treated
separately in Section~\ref{sec:populations}.

\section{Observation-law maps and the exact refinement law}
\label{sec:observationlaws}

The overall construction is summarized in Figure~\ref{fig:framework}: conformers
pass through physically explicit observable maps to observation-law maps that return
measurement distributions, from which we build both an exact quotient and a
finite-resolution ambiguity graph. The schematic separates the exact quotient
from the finite-resolution graph:
distinguishability is defined by the measurement model, and multimodality
refines that model rather than guaranteeing additive resolution.

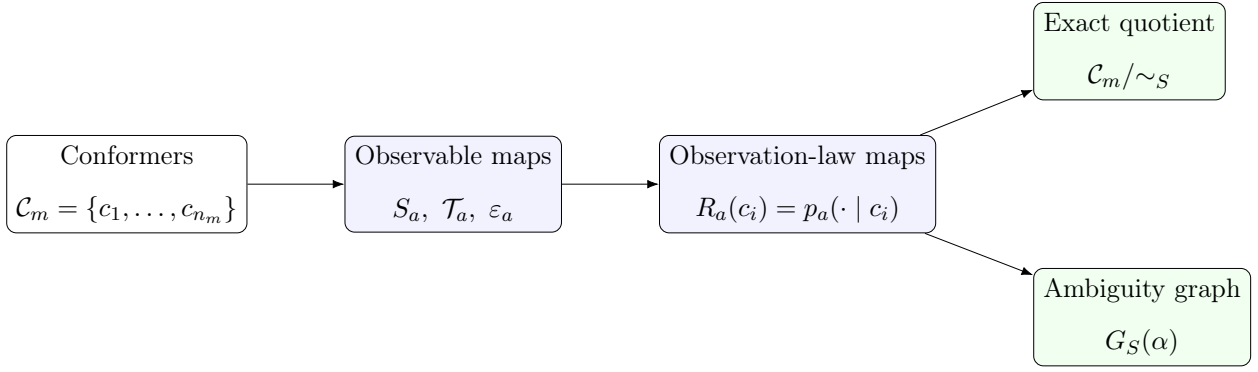
\begin{figure}[t]
\centering
\resizebox{\linewidth}{!}{%
\begin{tikzpicture}[
  >=Latex, node distance=7mm and 14mm, font=\small,
  box/.style={draw, rounded corners, align=center, inner sep=4pt, minimum height=9mm},
  obs/.style={box, fill=blue!5},
  outc/.style={box, fill=green!6}
]
\node[box] (conf) {Conformers\\$\Cset_m=\{c_1,\dots,c_{n_m}\}$};
\node[obs, right=of conf] (maps) {Observable maps\\$S_a,\ \Tmap_a,\ \varepsilon_a$};
\node[obs, right=of maps] (rec) {Observation-law maps\\$R_a(c_i)=p_a(\cdot\mid c_i)$};
\node[outc, above right=5mm and 14mm of rec] (quot) {Exact quotient\\$\Cset_m/{\sim_S}$};
\node[outc, below right=5mm and 14mm of rec] (graph) {Ambiguity graph\\$G_S(\alpha)$};
\draw[->] (conf) -- (maps);
\draw[->] (maps) -- (rec);
\draw[->] (rec) -- (quot);
\draw[->] (rec) -- (graph);
\end{tikzpicture}%
}
\caption{Framework schematic. Conformers pass through physically explicit
observable maps to observation-law maps that return measurement distributions; exact
observational equality yields a quotient, while finite-resolution ambiguity
yields a graph. The ambiguity graph is constructed under the declared
shared-covariance Gaussian working model, not the exact marginalized law
(Section~\ref{sec:finiteresolution}). Adding a modality refines the exact quotient
(Theorem~\ref{thm:refinement}) and can only remove ambiguity edges
[Eq.~\eqref{eq:edge-monotonicity}]; finite-noise ambiguity is a graph, not a
partition (Proposition~\ref{prop:nontransitive}).}
\label{fig:framework}
\end{figure}

\subsection{Observation-law maps and induced measurement distributions}

The central object is not an ideal spectrum but the distribution of
measurements it induces.

\begin{definition}[Observation-law map]
The observation-law map for modality $a$ maps each conformer to its
measurement distribution,
\begin{equation}
R_a(c_i)=p_a(\,\cdot\mid c_i),
\end{equation}
the law on $\Rbb^{d_a}$ of the observation $Y_a$ given $C=c_i$, with nuisance
parameters marginalized as in~\eqref{eq:marginal}.
\end{definition}

\begin{definition}[Exact observational equivalence]
Conformers $c_i,c_j$ are observationally equivalent under modality $a$,
written $c_i\sim_a c_j$, when their observation-law maps coincide as probability
measures:
\begin{equation}
c_i\sim_a c_j
\ \Longleftrightarrow\
R_a(c_i)=R_a(c_j)
\ \Longleftrightarrow\
p_a(\cdot\mid c_i)=p_a(\cdot\mid c_j)\ \text{a.e.},
\end{equation}
Equality of the densities almost everywhere is equivalent to equality of the
induced laws on $\Rbb^{d_a}$; we use the two interchangeably.
\end{definition}

Because $\sim_a$ is defined by equality of probability laws, it is
automatically reflexive, symmetric, and transitive. It therefore partitions
$\Cset_m$ into exact observational classes, denoted $\Cset_m/{\sim_a}$. This
quotient is the appropriate exact object: two conformers in the same class
cannot be distinguished by modality $a$, even with arbitrarily many repeated
measurements governed by the same observation law.

\paragraph{Running chemical example: 1,2-difluoroethane.}
The three-state 1,2-difluoroethane ensemble used below consists of an
\emph{anti} conformer and a declared mirror pair,
\emph{gauche}$^{+}$/\emph{gauche}$^{-}$. For the reflection-invariant achiral
observables used in this study, the two \emph{gauche} partners have identical
symmetrized observation laws, whereas the \emph{anti} law is distinct. For an
achiral modality set $S$, the exact quotient therefore has the form
\begin{equation}
\Cset_{\mathrm{DFE}}/{\sim_S}
=
\bigl\{\{c_{\mathrm{anti}}\},
       \{c_{g^+},c_{g^-}\}\bigr\}.
\label{eq:dfe-running-quotient}
\end{equation}
This example will recur below: finite resolution can make the
\emph{anti}/\emph{gauche} distinction easier or harder, but no combination of
the declared achiral channels can split the exact mirror class. A chiral or
parity-sensitive observable would be required to do so.

\subsection{Combining modalities}

For a set $S$ of modalities, write $y_S=(y_a)_{a\in S}$ for the joint
observation and define the combined observation-law map
\begin{equation}
R_S(c_i)=p(y_S\mid c_i).
\end{equation}
The channels are frequently conditionally independent given the
conformer: once $c_i$ is fixed, the physical noise processes of an infrared
spectrometer, a Raman spectrometer, and a microwave spectrometer are
independent, so
\begin{equation}
p(y_S\mid c_i)=\prod_{a\in S}p_a(y_a\mid c_i).
\label{eq:condindep}
\end{equation}
We flag~\eqref{eq:condindep} as an assumption rather than a law: shared
theoretical model error (for instance a common density-functional bias)
correlates the channels, and Section~\ref{sec:noise} describes how such
correlations are absorbed into a joint covariance when they cannot be
neglected.

The joint exact equivalence $\sim_S$ is defined exactly as before, with
$p_S$ in place of $p_a$.

\begin{theorem}[Refinement law]
\label{thm:refinement}
For any modality set $S$ and any additional modality $b$,
\begin{equation}
{\sim_{S\cup\{b\}}}\ \subseteq\ {\sim_S}\cap{\sim_b}.
\label{eq:refine-incl}
\end{equation}
If the joint observations factor as in~\eqref{eq:condindep}, then equality
holds:
\begin{equation}
{\sim_{S\cup\{b\}}}\ =\ {\sim_S}\cap{\sim_b}.
\label{eq:refine-eq}
\end{equation}
\end{theorem}

\begin{proof}
\emph{Inclusion~\eqref{eq:refine-incl}.} Suppose $c_i\sim_{S\cup\{b\}}c_j$, i.e.
$p_{S\cup\{b\}}(y_S,y_b\mid c_i)=p_{S\cup\{b\}}(y_S,y_b\mid c_j)$ for all
$(y_S,y_b)$. Integrating both sides over $y_b$ gives the marginal
$p_S(y_S\mid c_i)=p_S(y_S\mid c_j)$ for all $y_S$, so $c_i\sim_S c_j$;
integrating over $y_S$ gives $c_i\sim_b c_j$. Hence
$c_i\sim_{S\cup\{b\}}c_j$ implies both $c_i\sim_S c_j$ and $c_i\sim_b c_j$.
This step uses only marginalization and holds without any independence
assumption.

\emph{Equality~\eqref{eq:refine-eq} under~\eqref{eq:condindep}.} Suppose
$c_i\sim_S c_j$ and $c_i\sim_b c_j$. Using the factorization,
\begin{equation}
\begin{aligned}
p_{S\cup\{b\}}(y_S,y_b\mid c_i)
&=p_S(y_S\mid c_i)\,p_b(y_b\mid c_i)\\
&=p_S(y_S\mid c_j)\,p_b(y_b\mid c_j)
=p_{S\cup\{b\}}(y_S,y_b\mid c_j),
\end{aligned}
\end{equation}
for all $(y_S,y_b)$, so $c_i\sim_{S\cup\{b\}}c_j$. Combined with the inclusion,
this gives~\eqref{eq:refine-eq}.
\end{proof}

Equation~\eqref{eq:refine-incl} also has a direct partition
interpretation: every class of $\Cset_m/{\sim_{S\cup\{b\}}}$ is contained in a
class of $\Cset_m/{\sim_S}$. Adding a modality may therefore split an existing
observational class, but it cannot merge classes that were already distinct.
This is the exact, zero-noise sense in which additional spectroscopy cannot
hurt. In the 1,2-difluoroethane example of
Eq.~\eqref{eq:dfe-running-quotient}, another achiral modality may strengthen the
separation of \emph{anti} from the \emph{gauche} family, but it cannot split the
\emph{gauche}$^{+}$/\emph{gauche}$^{-}$ class because those partners remain
identical under every declared achiral observation law.

\section{Finite-resolution distinguishability}
\label{sec:finiteresolution}

Exact equality of measurement distributions is too strict for real
spectroscopy: two conformers whose distributions differ by an amount smaller
than the noise are, operationally, not distinguishable. We therefore replace
exact equivalence by a decision-theoretic separation.

\subsection{Pairwise Bayes error for Gaussian observations}

Consider deciding between the two hypotheses $C=c_i$ and $C=c_j$ from a single
observation $y_a$ under the Gaussian model~\eqref{eq:gaussmodel} with equal
priors. Equal priors are used because the graph is a pairwise
measurement-resolution diagnostic, not a thermodynamic-population classifier;
unequal priors would mix distinguishability with assumed conformer abundances.
When the two conformers share a covariance $\Sig_a$, define the squared
Mahalanobis separation
\begin{equation}
d_a^2(i,j)=\bigl(\muvec_{a,i}-\muvec_{a,j}\bigr)^{\Trans}
\Sig_a^{-1}\bigl(\muvec_{a,i}-\muvec_{a,j}\bigr).
\label{eq:mahalanobis}
\end{equation}
The Bayes-optimal decision rule is the likelihood-ratio (here, linear) test,
and its error probability is
\begin{equation}
P_{e,a}(i,j)=\Phi\!\left(-\tfrac12 d_a(i,j)\right),
\label{eq:bayeserror}
\end{equation}
where $\Phi$ is the standard-normal cumulative distribution function. This is
the standard two-class Gaussian result~\cite{Fukunaga1990,Kay1993}, and it is the
exact Bayes error for the declared shared-covariance Gaussian working
model of level (3) in Section~\ref{sec:observables}, not for the exact
marginalized observation law~\eqref{eq:marginal}, which is a non-Gaussian mixture.
Equation~\eqref{eq:bayeserror} is therefore a model-conditional decision error
under the declared Gaussian working model and covariance, rather than a
model-free physical quantity. It converts a spectral separation into a
probability and answers questions of the operational form:
\begin{quote}
At $8\wavenum$ resolution and the estimated model-error level, are $c_i$ and $c_j$ distinguishable with fewer than $5\%$ decision errors?
\end{quote}
Since $P_{e,a}(i,j)<0.05$ is equivalent to $d_a(i,j)>2\,\Phi^{-1}(0.95)\approx
3.29$, the question reduces to a threshold on the Mahalanobis separation.

The derivation is the standard equal-prior likelihood-ratio calculation: the
log-likelihood difference is a linear statistic with variance $d_a^2(i,j)$
under either hypothesis, and~\eqref{eq:bayeserror} follows by symmetry.

\subsection{Additivity and monotonicity}

Under conditional independence, the joint covariance is block diagonal,
$\Sig_S=\bigoplus_{a\in S}\Sig_a$, and the joint mean difference is the
concatenation of the per-modality differences. The Mahalanobis quadratic form
therefore separates into nonnegative channel contributions,
\begin{equation}
d_S^2(i,j)=\sum_{a\in S}d_a^2(i,j).
\label{eq:additive}
\end{equation}
Thus an independent modality contributes exactly its own squared separation; a
channel with no contrast for a particular pair contributes zero, while an
informative channel increases the combined separation.

A stronger ordering statement does not require Gaussianity or conditional
independence. Any decision rule based only on $y_S$ remains available after the
observation is enlarged to $(y_S,y_b)$: the rule can simply ignore $y_b$.
Consequently, the Bayes-optimal pairwise error satisfies
\begin{equation}
P_{e,S\cup\{b\}}(i,j)\ \le\ P_{e,S}(i,j),
\label{eq:errmonotone}
\end{equation}
for every added modality $b$. Equation~\eqref{eq:errmonotone} is the finite-noise
counterpart of the exact refinement law in Theorem~\ref{thm:refinement}. Under
the shared-covariance Gaussian model it also follows from
Eq.~\eqref{eq:additive}, because $d^2$ can only increase and
$\Phi(-\tfrac12 d)$ decreases with $d$.

\subsection{Scope of the shared-covariance approximation}

All ambiguity graphs reported below use the exact equal-prior Bayes error for
the declared shared-covariance Gaussian working model. When the covariance
depends materially on the conformer, the Bayes error must instead be evaluated
numerically. A Bhattacharyya-distance calculation can provide a one-sided
certificate that a pair is resolved~\cite{Kailath1967,Fukunaga1990}, but failure
of that certificate does not prove ambiguity. The corresponding formulas and logical qualifications are
given in the supplementary material.

\section{Ambiguity graphs and the failure of transitivity}
\label{sec:ambiguitygraph}

\subsection{Finite-noise ambiguity is not an equivalence relation}

It is tempting to define an operational ``indistinguishability'' by thresholding
the decision error,
\begin{equation}
c_i\approx c_j\quad\Longleftrightarrow\quad P_{e,S}(i,j)>\alpha,
\label{eq:naiverule}
\end{equation}
and to treat the resulting relation as if it partitioned $\Cset_m$. This is a
mathematical error: the relation~\eqref{eq:naiverule} is reflexive and
symmetric but not transitive.

\begin{proposition}[Non-transitivity of finite-resolution ambiguity]
\label{prop:nontransitive}
There exist configurations of conformer means for which $c_1\approx c_2$ and
$c_2\approx c_3$ but $c_1\not\approx c_3$ under the rule~\eqref{eq:naiverule}.
\end{proposition}

\begin{proof}
Take a single scalar observable ($d_a=1$) with unit variance and equal
covariance, and place three conformer means at $\mu_1=0$, $\mu_2=\Delta$, and
$\mu_3=2\Delta$ with $\Delta=2$. Then $d(1,2)=d(2,3)=2$ and $d(1,3)=4$,
so by~\eqref{eq:bayeserror}
\begin{equation}
P_e(1,2)=P_e(2,3)=\Phi(-1)\approx 0.159,\qquad
P_e(1,3)=\Phi(-2)\approx 0.023.
\end{equation}
Choosing the threshold $\alpha=0.05$ gives $c_1\approx c_2$ and $c_2\approx
c_3$ (both errors exceed $0.05$) while $c_1\not\approx c_3$ (its error is below
$0.05$). Transitivity fails.
\end{proof}

This non-transitivity has a direct representational consequence:
\begin{itemize}
\item exact observational equality is transitive and produces a quotient
(Section~\ref{sec:observationlaws});
\item finite-noise ambiguity is generally not transitive and must be
represented as a graph, not a partition.
\end{itemize}
Forcing~\eqref{eq:naiverule} into an equivalence relation, for example by
taking transitive closure, would incorrectly merge resolvable conformers and
would overstate ambiguity.

\subsection{The ambiguity graph and its observables}

\begin{definition}[Ambiguity graph]
For a modality set $S$ and tolerance $\alpha\in(0,\tfrac12)$, the
ambiguity graph is
\begin{equation}
G_S(\alpha)=(V,E_S),\qquad V=\Cset_m,\qquad
E_S=\bigl\{\{i,j\}:1\le i<j\le n_m,\ P_{e,S}(i,j)>\alpha\bigr\}.
\end{equation}
Here $P_{e,S}(i,j)$ is the Bayes error evaluated exactly under the declared
shared-covariance Gaussian model~\eqref{eq:bayeserror} or numerically for
conformer-dependent covariance; an edge therefore certifies genuine ambiguity
($P_e>\alpha$). When only a one-sided error bound is available, we report a
graph of pairs not certified as resolved rather than calling it an ambiguity
graph; details are given in the supplementary material.
\end{definition}

Combining the graph definition with Eq.~\eqref{eq:errmonotone} gives
\begin{equation}
E_{S\cup\{b\}}\subseteq E_S.
\label{eq:edge-monotonicity}
\end{equation}
Thus adding a modality can remove finite-noise ambiguity edges but cannot create
new ones under the same candidate set, priors, and observation model. If a pair
is already resolved using $S$, it remains resolved when $b$ is added because the
combined decision rule can always ignore the new channel.

Two graph-level observables summarize an ambiguity graph. The
unresolved-pair fraction
\begin{equation}
U_S(\alpha)=\frac{2\,|E_S|}{n_m(n_m-1)},
\label{eq:unresolved}
\end{equation}
is the fraction of conformer pairs that remain operationally ambiguous, and the
confusability degree
\begin{equation}
a_i(S,\alpha)=\sum_{j\ne i}\mathbf 1\!\left[P_{e,S}(i,j)>\alpha\right],
\label{eq:confusability}
\end{equation}
counts the competitors still confused with conformer $i$. By
Eq.~\eqref{eq:edge-monotonicity}, both are non-increasing as modalities
are added. These summaries let results be phrased in interpretable terms (for
example, ``modality $A$ leaves a fraction $U_A(\alpha)$ of low-energy pairs
ambiguous; adding modality $B$ reduces this to $U_{A\cup B}(\alpha)$'') rather
than as a single classifier accuracy; the specific fractions for the three
molecules are reported in Section~\ref{sec:pilot}.

Because $G_S(\alpha)$ need not be a disjoint union of cliques, its connected
components are not equivalence classes and should not be reported as such; we
report the raw graph and treat any clustering as a derived, $\alpha$-dependent
summary.

\paragraph{Zero-noise limit.}
Under the shared-covariance model~\eqref{eq:bayeserror} with fixed nuisance
parameters (or, under profiling, provided the transformed template manifolds
remain disjoint), contracting the observation covariances to zero
sends $d_S(i,j)\to\infty$ for every pair with $\muvec_{S,i}\ne\muvec_{S,j}$;
pairs whose induced means (or profiled template manifolds) coincide retain
error $\tfrac12$. The edge set of
$G_S(\alpha)$ thus converges to the exactly observationally equivalent pairs: the
finite-resolution graph degenerates to the exact quotient of
Section~\ref{sec:observationlaws}, which is thus the noiseless limit of the
finite-noise theory.

\section{Conformer mixtures and population identifiability}
\label{sec:populations}

Most ordinary spectra are not measurements of a single isolated conformer; they
are mixtures of conformer contributions. This section states criteria that are
logically prior to any inversion: can the chosen additive spectra determine the
populations at all? The three-molecule study below reports pure-conformer
ambiguity graphs and a rank demonstration for additive achiral IR/Raman columns.
The criteria apply to calibrated, unnormalized additive spectra; nonadditive
features such as $(A,B,C)$ and $f_{\bm\mu}$ are excluded from population columns.
Separate populations of rapidly interconverting mirror partners are physically
meaningful only under kinetic preparation, a chiral environment, or a
parity-sensitive measurement.

\subsection{Linear mixture model}

Let the conformer populations be
\begin{equation}
w_i\ge 0,\qquad \sum_{i=1}^{n_m}w_i=1,
\end{equation}
collected in $\wvec\in\Delta^{n_m-1}$, the probability simplex. The additive
mixture model
\begin{equation}
\yvec_a=M_a\wvec+\varepsilon_a,\qquad
M_a=\bigl[\,\muvec_{a,1}\ \muvec_{a,2}\ \cdots\ \muvec_{a,n_m}\,\bigr],
\label{eq:mixture}
\end{equation}
holds only for observables whose measured signal is the population-weighted
sum of the pure-conformer signals. This is the case for genuinely
additive, unnormalized intensity measurements:
\begin{itemize}
\item unnormalized IR absorbance (dilute Beer--Lambert response);
\item unnormalized Raman scattering intensity at fixed excitation and
temperature;
\item the complete simulated rotational line spectrum, whose transition
intensities are population-weighted and depend on the conformer dipole
components, degeneracies, and temperature.
\end{itemize}
Accordingly, the columns of $M_a\in\Rbb^{d_a\times n_m}$ are the unnormalized
conformer signals. Equation~\eqref{eq:mixture} assumes an externally calibrated
absolute intensity scale; the unknown-scale case is given below. The nonlinear
normalization used for pure-conformer discrimination is not applied here. The
model does not apply to the rotational-constant tuple $(A_i,B_i,C_i)$ or to the
permanent-dipole feature $f_{\bm\mu}(c_i)$: a mixture does not present a
population-weighted average of moments of inertia or of dipole magnitudes but a
superposition of transitions. We therefore use $(A,B,C)$ and $f_{\bm\mu}$ only
for pure-conformer discrimination
(Sections~\ref{sec:finiteresolution}--\ref{sec:design}); when they inform
mixtures it is through the additive rotational line spectrum above, not through
\eqref{eq:mixture} applied to the constants themselves. For a set $S$ of
additive modalities, stack the channels,
\begin{equation}
M_S=\begin{bmatrix}M_{a_1}\\ M_{a_2}\\ \vdots\end{bmatrix},
\qquad \yvec_S=M_S\wvec+\varepsilon_S.
\end{equation}

If the experiment has an unknown positive global intensity scale $a$, the mean
model is instead
\begin{equation}
\yvec_S=aM_S\wvec+\varepsilon_S.
\label{eq:mixture-scale}
\end{equation}
Writing $\bm p=a\wvec$ shows that this is a nonnegative linear inverse problem
with no sum constraint on $\bm p$. For interior mixtures, joint identification
of $(a,\wvec)$ for all admissible values therefore requires
$\ker(M_S)=\{0\}$, or $\rank(M_S)=n_m$; after recovering $\bm p$, one obtains
$a=\one^{\Trans}\bm p$ and $\wvec=\bm p/a$. This condition is stronger than the
calibrated-scale criterion below. Thus an unmodeled scale cannot be absorbed
into the simplex constraint without changing identifiability.

\subsection{Exact identifiability}

Two population vectors $\wvec,\wvec'$ produce the same mean observation exactly
when $M_S(\wvec-\wvec')=0$. Because both lie on the simplex,
$\one^{\Trans}(\wvec-\wvec')=0$, so the admissible difference vectors lie in the
simplex tangent space
\begin{equation}
H:=\{h\in\Rbb^{n_m}:\one^{\Trans}h=0\}.
\end{equation}

\begin{theorem}[Population identifiability with calibrated scale]
\label{thm:identifiability}
The conformer populations are uniquely determined by the noiseless mean
observations of modality set $S$ if and only if
\begin{equation}
\ker(M_S)\cap H=\{0\},
\label{eq:idcond}
\end{equation}
equivalently
\begin{equation}
\rank\!\begin{pmatrix}M_S\\ \one^{\Trans}\end{pmatrix}=n_m.
\label{eq:idrank}
\end{equation}
\end{theorem}

The proof is a standard kernel argument on the simplex tangent:
admissible population differences lie in $H=\ker(\one^{\Trans})$, and
noiseless observations identify $\wvec$ iff $\ker(M_S)\cap H=\{0\}$,
equivalently~\eqref{eq:idrank}. The full sufficiency/necessity argument is
given in the supplementary material.

Condition~\eqref{eq:idrank} can fail even with noiseless data whenever the
conformer signal columns are linearly dependent, for instance when one
conformer's spectrum is a convex combination of others'. More commonly the
condition holds but the stacked matrix is ill-conditioned, so populations are
identifiable in principle yet poorly determined in practice; this is quantified
in Section~\ref{sec:fisher}.

\section{Fisher information for population recovery}
\label{sec:fisher}

\subsection{Constrained information on the simplex}

Under the Gaussian mixture model~\eqref{eq:mixture} with joint covariance
$\Sig_S$, the mean is linear in $\wvec$, so the Fisher information for the
populations is
\begin{equation}
\Fish_S=M_S^{\Trans}\Sig_S^{-1}M_S.
\label{eq:fisher}
\end{equation}
For fixed $\Sig_S$, differentiation of the Gaussian log-likelihood gives
\begin{align}
\frac{\partial\log p}{\partial\wvec}
&=M_S^{\Trans}\Sig_S^{-1}(\yvec_S-M_S\wvec),\\
\Exp\!\left[
\frac{\partial\log p}{\partial\wvec}
\left(\frac{\partial\log p}{\partial\wvec}\right)^{\Trans}
\right]
&=M_S^{\Trans}\Sig_S^{-1}M_S,
\end{align}
which is~\eqref{eq:fisher}.
Equation~\eqref{eq:fisher} assumes that $\Sig_S$ is positive definite, does not
depend on $\wvec$ (homoscedastic in the populations), and is known or
externally calibrated, and that $\wvec$ lies in the interior of the simplex so
the reduced coordinates below are unconstrained. If the noise is signal-dependent
(so $\Sig_S=\Sig_S(\wvec)$), an additional term involving
$\partial\Sig_S/\partial\wvec$ enters~\eqref{eq:fisher}; at the simplex boundary
the ordinary unconstrained Cram\'er--Rao interpretation requires the usual
qualification for parameters on the boundary of the feasible set. We state
results for the interior, fixed-$\Sig_S$ case and treat the heteroscedastic and
boundary cases as noted.
Population changes must preserve normalization, so only directions in the
tangent space $H$ are physical. Let the columns of $U\in\Rbb^{n_m\times(n_m-1)}$
form an orthonormal basis of $H$. The physically relevant constrained
Fisher information is
\begin{equation}
\Fish_S^{(\Delta)}=U^{\Trans}\Fish_S U.
\label{eq:constrainedfisher}
\end{equation}
Its inverse bounds the covariance of any unbiased estimator of the reduced
coordinates through the Cram\'er--Rao inequality~\cite{Kay1993}, so the smallest
eigenvalue $\lambda_{\min}(\Fish_S^{(\Delta)})$ measures the least-resolved
population contrast. If it is zero, some
population exchange is unidentifiable, consistent with the failure
of~\eqref{eq:idcond}; if it is small, the inverse problem is ill-conditioned
and the corresponding population contrast is recoverable only with large
uncertainty.

For the unknown-scale model~\eqref{eq:mixture-scale}, let
$Q=M_S^{\Trans}\Sig_S^{-1}M_S$. In local coordinates $(a,\bm\theta)$ with
$\wvec=\wvec_0+U\bm\theta$, the mean Jacobian is
$[\,M_S\wvec,\ aM_SU\,]$. Profiling the scalar scale from the joint Fisher
matrix gives the Schur complement
\begin{equation}
\Fish_{\wvec\mid a}^{(\Delta)}
=a^2U^{\Trans}QU
-\frac{a^2(U^{\Trans}Q\wvec)(\wvec^{\Trans}QU)}
{\wvec^{\Trans}Q\wvec},
\label{eq:profiled-scale-fisher}
\end{equation}
provided $\wvec^{\Trans}Q\wvec>0$. The subtracted positive-semidefinite term is
the information lost to the unknown scale. Equation~\eqref{eq:constrainedfisher}
is recovered when the absolute scale is externally calibrated (with $a=1$ in
the convention used here).

For conditioning and numerical-rank diagnostics we use the whitened constrained design
\begin{equation}
\widetilde{M}_S:=\Sig_S^{-1/2}M_S\,U,
\qquad
\Fish_S^{(\Delta)}=\widetilde{M}_S^{\Trans}\widetilde{M}_S,
\label{eq:whitened}
\end{equation}
whose singular values are the square roots of the eigenvalues of
$\Fish_S^{(\Delta)}$. Conditioning statements below refer to $\widetilde{M}_S$,
not to the raw stacked $M_S$: the condition number of $M_S$ alone depends on the
arbitrary physical units and relative scaling assigned to the IR, Raman,
and other additive spectral channels, whereas whitening by $\Sig_S^{-1/2}$ and
restricting to the simplex tangent removes that dependence and yields the
statistically meaningful conditioning of the population inverse problem.

\subsection{Additive information law}

If the measurement noise of modality $b$ is independent of that of $S$
conditional on the population vector, then
$\Sig_{S\cup\{b\}}=\Sig_S\oplus\Sig_b$ and
\begin{equation}
\Fish_{S\cup\{b\}}
=
\Fish_S+M_b^{\Trans}\Sig_b^{-1}M_b.
\label{eq:fisheradd}
\end{equation}
The added term is positive semidefinite because it is a Gram matrix.
Consequently,
\begin{equation}
\Fish_{S\cup\{b\}}-\Fish_S\succeq0,
\qquad
\Fish_{S\cup\{b\}}^{(\Delta)}
-\Fish_S^{(\Delta)}\succeq0.
\label{eq:fisherpsd}
\end{equation}
Thus $\lambda_{\min}(\Fish^{(\Delta)})$ and
$\log\det[\Fish^{(\Delta)}+\lambda I]$ cannot decrease when an independent
additive measurement is added. The block-matrix derivation is given in the
supplementary material.

This is an information-refinement result for population recovery: an
additional independent spectroscopic measurement cannot decrease the Fisher
information, and it strictly increases it in any direction to which the new
modality is sensitive. It is the quantitative, mixture-level analogue of the
refinement law of Theorem~\ref{thm:refinement}.

\section{Optimal selection of the next observable}
\label{sec:design}

The refinement and information laws are monotone, but each additional modality
or spectral window costs acquisition effort. The framework therefore ranks
which measurement to perform next.

\subsection{Objectives for pure-conformer discrimination}

For discriminating single conformers, two natural objectives are to improve the
worst-resolved pair,
\begin{equation}
b^\star=\argmax_b\ \min_{i\ne j} d^2_{S\cup\{b\}}(i,j),
\label{eq:maxmin}
\end{equation}
or to remove the largest number of ambiguity edges,
\begin{equation}
b^\star=\argmin_b\ U_{S\cup\{b\}}(\alpha).
\label{eq:minU}
\end{equation}
Objective~\eqref{eq:maxmin} is an E-optimal-type criterion on pairwise
separations; \eqref{eq:minU} is a direct combinatorial reduction of ambiguity.
Objective~\eqref{eq:maxmin} becomes uninformative whenever the conformer set
contains a symmetry-forbidden pair (an exactly observationally equivalent
mirror pair, Section~\ref{sec:pilot}) because its separation is zero for
every achiral candidate modality, so the minimum over pairs is pinned at
zero regardless of $b$. In that case the max--min objective must either be
restricted to the pairs that are not symmetry-forbidden or extended to include a
parity-sensitive (chiroptical) candidate that can lift the exact degeneracy.
The numerical study below takes the restricted option: every declared mirror
pair is symmetrized before graph construction, so max--min is evaluated on the
mirror-collapsed quotient, where it is not identically zero
(Section~\ref{sec:pilot}). The
\eqref{eq:minU} objective is unaffected because a permanently unresolved pair contributes a
constant edge to every candidate's $U_{S\cup\{b\}}$.

\subsection{Objectives for population recovery}

For population recovery, the same selection logic applies to scalar summaries of
the constrained information~\eqref{eq:constrainedfisher}. With a ridge
$\lambda I$ that regularizes near-unidentifiable cases, the standard alphabetic
and Bayesian optimal-design
criteria~\cite{Fedorov1972,Pukelsheim2006,ChalonerVerdinelli1995} specialize to
\begin{align}
b^\star_{\mathrm D}&=\argmax_b\ \log\det\!\left[\Fish^{(\Delta)}_{S\cup\{b\}}
+\lambda I\right],\\
b^\star_{\mathrm E}&=\argmax_b\ \lambda_{\min}\!\left[
\Fish^{(\Delta)}_{S\cup\{b\}}\right],
\end{align}
applied only to additive population channels admitted by
Section~\ref{sec:populations}.

\subsection{Candidate modalities and spectral windows}

The candidate set need not consist of whole modalities. It may instead contain
spectral windows or experimentally linked feature packages. In the present
study the candidates are evaluated exhaustively, using either the reduction in
unresolved-pair fraction or the improvement in the least-resolved target pair.
The tested candidates include:
\begin{itemize}
\item the IR fingerprint region;
\item the low-frequency (far-IR/terahertz) region, sensitive to torsions;
\item the O--H/N--H stretching region, sensitive to hydrogen bonding;
\item the Raman fingerprint region;
\item the low-frequency Raman region;
\item rotational constants $(A,B,C)$ and permanent-dipole features for
pure-conformer discrimination only;
\item a complete, unnormalized rotational line spectrum for population design;
\item selected XAS windows for heteroatom-rich systems.
\end{itemize}
Restricting the design to windows identifies not merely which spectroscopy is
useful but where the conformational information resides. Population-design
objectives use only additive signals admitted by Section~\ref{sec:populations};
neither $(A,B,C)$ nor $f_{\bm\mu}$ is a population-mixture channel.

\section{Computational protocol}
\label{sec:protocol}

This section documents the computational pipeline used in the study of
Section~\ref{sec:pilot}. It uses only open-source electronic-structure software,
and the downstream DFT, measurement-operator, and analysis stages are fixed by a
single configuration object. The CREST search is stochastic: its metadynamics
seed was not fixed, so the archived downstream analysis is reproducible whereas
a new conformer search is a statistical replication. Versions, thresholds, and
measurement settings are collected in Table~S1 of the
supplementary material. The processed artifacts, configurations, analysis code,
and archived raw conformer-search and electronic-structure records are available
in the public Zenodo deposit described in the Data Availability statement. The
workflow is summarized in Figure~\ref{fig:workflow}.

\begin{figure}[t]
\centering
\resizebox{\linewidth}{!}{%
\begin{tikzpicture}[
  >=Latex, node distance=5.5mm and 7mm, font=\footnotesize,
  stage/.style={draw, rounded corners, align=center, inner sep=4pt,
    minimum height=12mm, text width=34mm},
  gen/.style={stage, fill=blue!5},
  calc/.style={stage, fill=orange!8},
  ana/.style={stage, fill=green!7},
  outbox/.style={stage, fill=gray!10}
]
\node[gen] (crest) {CREST / GFN2-xTB\\conformer sampling};
\node[gen, right=of crest] (filter) {De-duplicate;\\topology filter};
\node[gen, right=of filter] (dft) {Dispersion-corrected\\DFT re-optimization};
\node[gen, right=of dft] (clust) {Post-DFT symmetry\\and torsion audit};
\node[calc, below=10mm of clust] (spectra) {NWChem spectra:\\IR, Raman, $(A,B,C)$, $f_{\bm\mu}$, $G$};
\node[ana, left=of spectra] (meas) {Measurement $\Tmap_a$;\\noise model $\Sigma_S$};
\node[ana, left=of meas] (dist) {Distinguishability\\analysis:\\$d^2$, $P_e$, $G_S(\alpha)$};
\node[outbox, left=of dist] (outp) {Ambiguity graphs\\and design};
\draw[->] (crest) -- (filter);
\draw[->] (filter) -- (dft);
\draw[->] (dft) -- (clust);
\draw[->] (clust) -- (spectra);
\draw[->] (spectra) -- (meas);
\draw[->] (meas) -- (dist);
\draw[->] (dist) -- (outp);
\end{tikzpicture}%
}
\caption{Computational workflow. Conformers are generated with CREST on the
GFN2-xTB surface~\cite{Grimme2019,Pracht2020,Bannwarth2019}, filtered and
re-optimized at a dispersion-corrected DFT level, and audited using principal
torsions, molecular symmetry, and mirror relationships. Harmonic IR and Raman spectra,
rotational constants, dipoles, and free energies are computed with
NWChem~\cite{NWChem2020}; the measurement operator and noise/model-error
covariance (Sections~\ref{sec:observables},~\ref{sec:noise}) convert ideal
spectra into measurement distributions; and the framework produces ambiguity,
population-identifiability, and design diagnostics.}
\label{fig:workflow}
\end{figure}
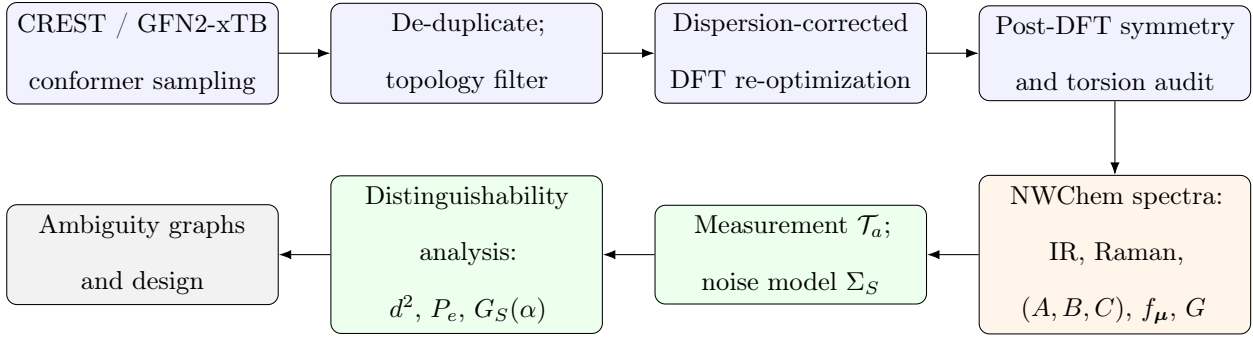

\subsection{Molecule set and curation}
\label{sec:dataset}

The study of Section~\ref{sec:pilot} applied the pipeline to three molecules
chosen to span a difficulty gradient: 1,2-difluoroethane, ethylene glycol, and
\emph{n}-pentane. All are neutral, closed-shell molecules with low-energy
rotamers.

\paragraph{Curation workflow.} (i) Generate conformers with CREST using the
GFN2-xTB tight-binding surface and metadynamics
sampling~\cite{Grimme2019,Pracht2020,Bannwarth2019}. (ii) Remove
topology-changing structures, identified from a covalent-bond graph built with
Cordero covalent radii~\cite{Cordero2008}, and duplicates. (iii) Reoptimize at a consistent
dispersion-corrected DFT level. (iv) Apply an initial fixed-order Kabsch RMSD
filter~\cite{Kabsch1976}, then audit principal torsions, molecular symmetry, and mirror
relationships after DFT optimization. (v) Retain thermally relevant structures by a
free-energy cutoff. This route was used here, but the initial search remains
stochastic. The initial fixed-order RMSD does not account for all permutations
of symmetry-equivalent atoms, so the final audit uses principal torsions,
molecular end reversal, proper-rotation Kabsch checks, and an explicit
distinction between proper symmetry matches and mirror relations
(Section~\ref{sec:pilot-settings}).

\subsection{Spectroscopic calculations}
\label{sec:spectra}

All DFT stages (the dispersion-corrected geometry optimizations of both the
screen and the production level, the cross-method re-optimizations, and the
harmonic frequency, IR, and Raman calculations) were carried out with the same
open-source NWChem package~\cite{NWChem2020}; CREST/GFN2-xTB was used only for the
initial conformer search. For each retained conformer NWChem yielded:
\begin{itemize}
\item harmonic IR spectra (frequencies and intensities);
\item harmonic Raman spectra (static-limit activities);
\item rotational constants $(A,B,C)$;
\item the permanent-dipole vector (reduced to the reflection-invariant feature
$f_{\bm\mu}$ of Eq.~\eqref{eq:dipolefeature});
\item conformer free energies (for Boltzmann weights).
\end{itemize}
Harmonic frequencies were multiplied by the working factor $s_{\mathrm{harm}}=0.975$
(Table~S1; cf.\ Merrick et
al.~\cite{MerrickMoranRadom2007}). Raman
activities $A_k$ were converted to observed intensities at a fixed
excitation wavenumber $\tilde\nu_{\mathrm{exc}}$ and temperature $T$ by the standard
scattering factor
\begin{equation}
I_k\ \propto\ \frac{(\tilde\nu_{\mathrm{exc}}-\nu_k)^4\,A_k}
{\nu_k\bigl[1-\exp(-hc\,\nu_k/k_BT)\bigr]},
\label{eq:ramanintensity}
\end{equation}
with the same $\tilde\nu_{\mathrm{exc}}$ and $T$ applied to every conformer~\cite{Long2002}. VCD is highly
conformer-sensitive and scientifically valuable but complicates both the
calculation and the software stack; we reserve it, together with XAS for
heteroatom-rich systems, for later work once the framework is established on the
less expensive vibrational and rotational observables.

NWChem supplies rigid-rotor/harmonic-oscillator (RRHO) thermochemistry. We tested
its low-frequency sensitivity with the quasi-RRHO entropy interpolation of
Grimme~\cite{Grimme2012}. For each positive vibrational frequency $\nu_k$,
\begin{align}
w_k&=\left[1+\left(\frac{\tilde\nu_{\mathrm{cut}}}{\nu_k}\right)^p\right]^{-1},
&S_k^{\rm qRRHO}&=w_kS_k^{\rm HO}+(1-w_k)S_k^{\rm FR},
\label{eq:qrrho}
\end{align}
with $\tilde\nu_{\mathrm{cut}}=100\wavenum$ and $p=4$; $S_k^{\rm FR}$ is
Grimme's one-dimensional free-rotor entropy with the reduced modal inertia
built from the isotropically averaged molecular moment of inertia
$B_{\mathrm{av}}=\tfrac13(I_a+I_b+I_c)$. The implemented expressions, including
the free-rotor formula and the modal-inertia definitions, are collected in the
supplementary material. We retained the NWChem RRHO enthalpy correction
and replaced only the vibrational entropy,
$G_{\rm qRRHO}=G_{\rm RRHO}-T(S_{\rm qRRHO}-S_{\rm HO})$, at $298.15\,$K.
Unscaled harmonic frequencies were used for thermochemistry; $s_{\mathrm{harm}}=0.975$
applies only to the displayed spectral measurement operator. The quasi-RRHO
test asks whether low modes change conformer retention or relative weights
under that interpolation.

\subsection{Cross-method evaluation}
\label{sec:inversecrime}

If the simulated observation and the candidate library are produced by exactly
the same calculation, apparent separation can simply reflect self-consistency.
The reported graphs used B3LYP means with the fixed working covariance of
Section~\ref{sec:noise}. A six-case cross-method pseudo-observation diagnostic on
the archived two-conformer-per-molecule subset appears in
Table~\ref{tab:reassign-summary}; experimental consistency checks appear in
Section~\ref{sec:expcompare}. How these relate to the working covariance is
stated once in Section~\ref{sec:discussion}.

\subsection{Noise and model-discrepancy study}
\label{sec:noise}

For a spectral modality the covariance of the processed signal $\muvec$ on the
wavenumber grid is built explicitly as
\begin{equation}
\Sig
=\underbrace{\sigma_{\mathrm{meas}}^2 I}_{\text{measurement noise}}
+\underbrace{(\sigma_{\mathrm m}\sigma_{\mathrm m}^{\Trans})\circ K_{\mathrm m}}_{\text{correlated model error}}
+\underbrace{\sigma_\nu^2\,(g g^{\Trans})\circ K_\nu}_{\text{local frequency-shift error}}
+\underbrace{\sigma_{\mathrm b}^2 K_{\mathrm b}}_{\text{baseline}},
\label{eq:covmodel}
\end{equation}
where $\sigma_{\mathrm{meas}}=\max_k|\mu_k|/\mathrm{SNR}$ is homoscedastic
measurement noise defined by the peak signal-to-noise ratio (SNR);
$\sigma_{\mathrm m}=\rho_{\mathrm m}(|\muvec|+c)$ with
$\rho_{\mathrm m}=0.05$ in the working covariance is a
signal-proportional (heteroscedastic) model-error amplitude correlated by a
squared-exponential kernel $K_{\mathrm m}$ of length $\ell_{\mathrm m}$
(a theoretical frequency/shape error moves adjacent bins together, which
reduces the effective discriminating information relative to i.i.d.\
noise); $g=\mathrm{d}\muvec/\mathrm{d}\nu$ and the rank-structured term
$\sigma_\nu^2(gg^{\Trans})\circ K_\nu$ encodes a local frequency-shift
uncertainty, the term that lets a genuine conformer peak shift be confused with
theoretical error; and $K_{\mathrm b}$ is an optional baseline kernel. Setting
$\ell_{\mathrm m}=0$ recovers a diagonal $\Sig$. Correlated and cross-modality
model-error components can populate the off-diagonal and cross-modality blocks
of $\Sig_S$, which is where the conditional-independence
assumption~\eqref{eq:condindep} is relaxed when warranted. The reported
three-molecule calculations used correlated within-spectrum blocks but a
block-diagonal stack across modalities. Robustness of the reported graphs was
assessed in Section~\ref{sec:pilot} by one-factor and combined stress-test
sweeps of full width at half maximum (FWHM), SNR, local frequency-shift
uncertainty, and model-intensity error about the working point of
Table~S1.

For reproducibility, the covariance used in
Table~\ref{tab:pilot-refinement} was constructed by a fixed algorithm: one
reference signal $\bar\mu_a$ was formed from the ensemble mean of the validated
production minima, Eq.~\eqref{eq:covmodel} was evaluated at $\bar\mu_a$ with the
fixed parameters of Table~S1 and held fixed across all
curation sweeps, $(A,B,C)$ and $f_{\bm\mu}$ carried diagonal covariances, and the
stack was block diagonal. The step-by-step construction is given in the
supplementary material.

The production calculations used
B3LYP-D3(BJ)/def2-TZVP~\cite{Becke1993,LeeYangParr1988,Stephens1994,GrimmeD32010,GrimmeBJ2011,WeigendAhlrichs2005},
with B3LYP-D3(BJ)/def2-SVP for screening and
PBE0-D3(BJ)/def2-TZVP~\cite{AdamoBarone1999} for the cross-method subset. Spectra were evaluated on a
$200$--$4000\wavenum$ grid with $8\wavenum$ Gaussian FWHM and the declared
frequency factor $s_{\mathrm{harm}}=0.975$. The working covariance used peak
SNR~50, model-error fraction $\rho_m=0.05$, model-error correlation length
$20\wavenum$, frequency-shift uncertainty $\sigma_\nu=3\wavenum$, and $5\%$
relative uncertainty for the rotational and dipole features. Quasi-RRHO
retention used $5\,\mathrm{kcal\,mol^{-1}}$
($\tilde\nu_{\mathrm{cut}}=100\wavenum$, $p=4$), and distinguishability is
declared at $\alpha=0.05$ with sweeps over $0.01$ and $0.10$. Complete
software, convergence, curation, and nuisance-parameter settings are given in
supplementary Table~S1.

\section{Application to three molecular ensembles}
\label{sec:pilot}

We now apply the framework end to end to three molecules chosen to span a
difficulty gradient: 1,2-difluoroethane (rigid, two rotamer families), ethylene
glycol (intramolecular hydrogen bonding), and \emph{n}-pentane (a floppy
hydrocarbon with near-degenerate rotamers). The complete pipeline (CREST
conformer search, DFT curation, NWChem spectra, the measurement operator, and
the pure-conformer distinguishability and design diagnostics) was run on
genuine electronic-structure output. Unless noted otherwise, qualitative conclusions use the
audited def2-TZVP geometries and the fixed working covariance
(Table~S1), with the dipole channel evaluated as the
reflection-invariant feature $f_{\bm\mu}$ of Eq.~\eqref{eq:dipolefeature}
(component magnitudes, not signed Cartesian entries). Mirror-related edges
survived in all three examples after symmetry/permutation duplicates were
collapsed.

\subsection{Validated molecular ensembles}
\label{sec:pilot-settings}

Five independent CREST searches followed by the DFT and symmetry audit produced
3, 20, and 7 unique screened representatives for 1,2-difluoroethane, ethylene
glycol, and \emph{n}-pentane, respectively. All representatives were promoted
to the production level. Three ethylene-glycol structures retained one
imaginary mode and were excluded as non-minima, leaving validated production
ensembles of 3, 17, and 7 conformers. A stratified subset of two conformers per
molecule was recomputed at PBE0-D3(BJ)/def2-TZVP for the cross-method
diagnostic, and production-level Raman activities were obtained for all seven
\emph{n}-pentane representatives.

RRHO and quasi-RRHO retention selected the same structures within
$5\,\mathrm{kcal\,mol^{-1}}$ for all three molecules. Across the complete
thermal-window, cap, torsion-threshold, decision-tolerance, and
RRHO/quasi-RRHO sweep, every residual full-stack edge remained a declared
mirror pair; the edge count changed only when the retained ensemble size
changed. The full sensitivity results and conformer-search diagnostics are
reported in the supplementary material.

Before retention, graph construction, and population analysis, each declared
mirror pair was symmetrized by averaging its thermochemical values and every
achiral observable column. The expanded inventory in the supplementary
material lists all validated minima, mirror relations, torsional descriptors,
and excluded imaginary-mode structures.

Three behaviors emerge from these ensembles and organize the results below:
Raman need not refine the IR graph; reflection-invariant $f_{\bm\mu}$ leaves
mirror-related edges; and incomplete symmetry or permutation handling creates
spurious clique structure.

\subsection{Refinement and the limits of achiral spectroscopy}

Table~\ref{tab:pilot-refinement} reports the unresolved-pair fraction for every
single modality and for the complete achiral stack under the primary uncapped
quasi-RRHO $5\,\mathrm{kcal\,mol^{-1}}$ ensemble, and
Figure~\ref{fig:pilot-graphs} shows the corresponding full-modality ambiguity
graphs. Several rows repeat the same value, and that repetition is the point:
under the working covariance no achiral combination improves on IR alone,
because every surviving edge is symmetry-protected, not because the added
modalities are uninformative in general. The intermediate combinations, which
all reproduce the IR row, are tabulated in the supplementary material.
Consistent with Eq.~\eqref{eq:edge-monotonicity}, no entry ever increases as
a modality is added: within each column $U_S$ is non-increasing along any chain
of nested modality sets. Figure~\ref{fig:pilot-spectra} makes the underlying
construction concrete for 1,2-difluoroethane, tracing one pair from ideal
harmonic lines through a broadened, noisy pseudo-observation to a Bayes-error
decision, and contrasting a resolved pair with the degenerate mirror pair whose
Bayes error is exactly $\tfrac12$.

Two qualifications apply to these numbers. First, the graphs
were computed within the B3LYP working model: means and covariance came
from the same electronic-structure method (Section~\ref{sec:noise}). The
statement that the calculated IR means separated all non-mirror pairs is therefore
a within-method separation of the illustrative ensembles, not evidence that IR is
generally sufficient to resolve non-mirror conformers; in particular it does not
survive the combined stress test, under which four non-mirror \emph{n}-pentane
quotient edges appear and are removed only by Raman, rotation, or $f_{\bm\mu}$
(Sections~\ref{sec:covariance-robustness} and~\ref{sec:window-design}). The only out-of-method
probe included here, the six-case cross-method pseudo-observation diagnostic
(Table~\ref{tab:reassign-summary}), returned the intended representative in one of
six cases under fixed calibration (rising to three of six representatives under
the working covariance, four of six counting assignment to the correct achiral
class, and five of six representatives with all six achiral classes correct
under the conservative cross-method covariance, once per-candidate calibration was
profiled; Section~\ref{sec:discussion}). We accordingly phrase
every headline outcome as a within-method separation under the declared
covariance.

Second, the residual full-stack edges in
Figure~\ref{fig:pilot-graphs} are exactly the declared mirror pairs, so their
count (1/8/3) is fixed by the symmetry construction rather than discovered
numerically: mirror partners satisfy $d=0$ and $P_e=\tfrac12$ because the energies and
achiral observable columns were symmetrized to be identical within each
declared mirror pair. The natural object for finite-noise measurement design is therefore
the graph on the exact observational quotient $\Cset_m/{\sim_S}$, in which
mirror partners have already been collapsed. On that quotient the working-covariance
graphs had no edges for all three molecules (the mirror floor is the whole
residual), so the achiral design problem is trivial under the working covariance
and becomes nontrivial only under the \emph{n}-pentane stress-test covariance
(Section~\ref{sec:covariance-robustness}), which introduces genuine non-mirror
quotient edges. Figure~\ref{fig:pilot-graphs} shows the node graph, which
exhibits the symmetry floor; the mirror-collapsed quotient graph, which isolates
the finite-resolution ambiguity that measurement can actually address, is
summarized compactly in Table~\ref{tab:quotient} (node and class counts and
quotient edges under IR and the full stack for the working, cross-method, and
combined stress-test covariances) rather than drawn as a separate figure. Only
one nontrivial quotient edge appears anywhere in the study, in stressed
\emph{n}-pentane, and the full stack removes it.

\begin{table}[htbp]
\centering
\begin{threeparttable}
\caption{Quotient (mirror-collapsed) ambiguity graph summary at $\alpha=0.05$. Each molecule's nodes collapse to exact achiral classes; a \emph{quotient edge} is a finite-noise ambiguity edge between two \emph{distinct} classes, so mirror-partner edges (the symmetry floor) never appear. Column groups are the working covariance, the conservative cross-method covariance, and the combined stress test (FWHM $16\wavenum$, SNR~$20$, $\sigma_\nu=12\wavenum$, $\rho_m=0.30$); \emph{full} is the complete achiral stack.}
\label{tab:quotient}
\footnotesize
\setlength{\tabcolsep}{4pt}
\begin{tabular}{@{}lcccccccc@{}}
\toprule
 & & & \multicolumn{2}{c}{working} & \multicolumn{2}{c}{cross-method} & \multicolumn{2}{c}{combined stress test}\\
\cmidrule(lr){4-5}\cmidrule(lr){6-7}\cmidrule(l){8-9}
Molecule & nodes & classes & IR & full & IR & full & IR & full\\
\midrule
1,2-difluoroethane & 3 & 2 & 0 & 0 & 0 & 0 & 0 & 0\\
ethylene glycol & 17 & 9 & 0 & 0 & 0 & 0 & 0 & 0\\
\emph{n}-pentane & 7 & 4 & 0 & 0 & 0 & 0 & 1 & 0\\
\bottomrule
\end{tabular}
\begin{tablenotes}
\footnotesize
\item The single nontrivial quotient edge in the study is the stressed \emph{n}-pentane pair ($g^{+}T/g^{-}T$--$g^{+}g^{+}/g^{-}g^{-}$), into which four finite-noise node edges collapse. It is removed by IR+Raman, IR+rotation, or IR+$f_{\bm\mu}$, so the full stress-test stack again has zero quotient edges.
\end{tablenotes}
\end{threeparttable}
\end{table}

\begin{figure}[htbp]
\centering
\includegraphics[width=\linewidth]{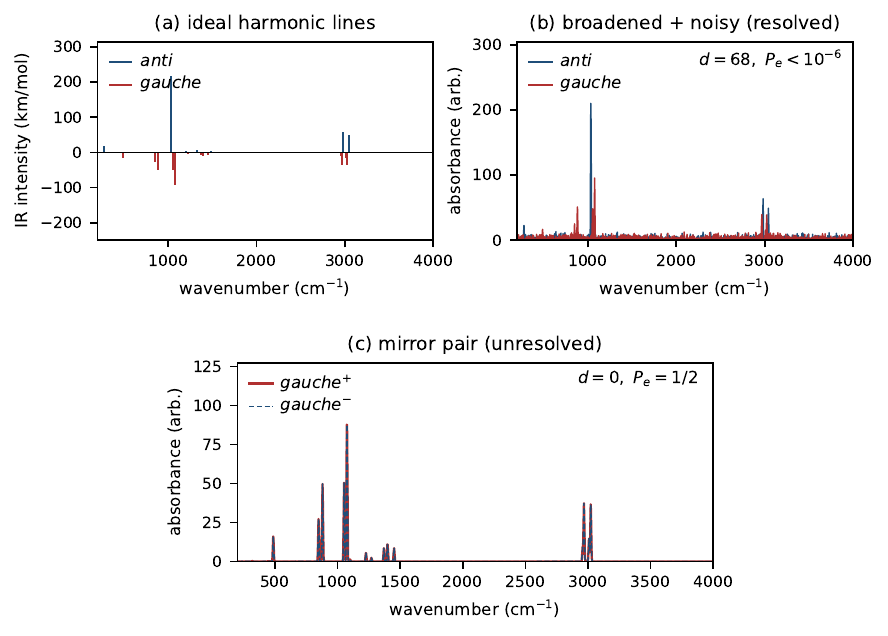}
\caption{From computed lines to a decision, for 1,2-difluoroethane at the
production B3LYP-D3(BJ)/def2-TZVP level. (a)~Ideal harmonic IR lines
(frequency-scaled by $0.975$) for the \emph{anti} conformer (up, blue) and a
\emph{gauche} conformer (down, red). (b)~The same pair after Gaussian
broadening ($8\wavenum$ FWHM) and an illustrative noisy draw (peak SNR $50$),
which the working covariance resolves. (c)~The mirror-image
\emph{gauche}$^{\pm}$ rotamers, which coincide exactly under the achiral map.}
\label{fig:pilot-spectra}
\end{figure}

For ethylene glycol the primary full stack retained eight edges among seventeen
validated minima ($U=0.059$), each a declared mirror pair after the audit and
imaginary-mode exclusions listed in the supplementary inventory. Rotational constants
left a much denser graph ($U=0.485$) under the declared $5\%$ matching
uncertainty, because many hydrogen-bonded rotamers have nearly coincident
moments of inertia. Within the working model, infrared and Raman each separated
all chemically distinct (non-mirror) pairs in this ensemble. The reflection-invariant dipole feature
$f_{\bm\mu}$ alone left ten edges ($U=0.074$): the eight mirror pairs plus two
additional chemically distinct pairs that remained ambiguous under the declared
dipole covariance. Adding Raman, rotation, or $f_{\bm\mu}$ to the IR baseline
nevertheless removed zero of the eight residual IR edges, because those edges
are mirror-protected. No candidate in the present achiral set can lower $U$
(Section~\ref{sec:design}); once those candidates are exhausted, the remaining
structural requirement is a parity-sensitive channel: VCD, Raman optical
activity (ROA), or microwave three-wave mixing. Those alternatives were not
ranked against one another.

For 1,2-difluoroethane, IR and Raman each left a single unresolved pair
($U=0.33$), the mirror-image \emph{gauche}$^{\pm}$ rotamers, and adding Raman
or rotation to IR removed zero edges (vibrational redundancy). The
reflection-invariant dipole feature likewise left that edge:
$U_{f_{\bm{\mu}}}=0.33$ and the full-modality graph retained one edge
(Figure~\ref{fig:pilot-graphs}a). The \emph{anti}/\emph{gauche} contrast was
resolved by IR and by $f_{\bm\mu}$ (the \emph{anti} conformer is apolar by
symmetry); what survived was exactly the achiral mirror floor. At the
def2-SVP screening level the same \emph{gauche}$^{\pm}$ pair already survived every
achiral modality; the production level did not lift it, nor should it under
Eq.~\eqref{eq:dipolefeature}.

For \emph{n}-pentane the audited seven-conformer set retained three full-stack
edges ($U=0.143$): the mirror pairs $g^\mp T$, $g^\mp g^\mp$, and $g^\pm x^\mp$
(Figure~\ref{fig:pilot-graphs}c). Adding rotation or the invariant dipole
feature to IR removed no edge. Archive records that would otherwise form an
all-\emph{trans} clique were symmetry/permutation duplicates; after collapse they
contributed a single self-achiral $TT$ node.

As noted below (Section~\ref{sec:expcompare}), the production ensemble
reproduced the experimentally established \emph{gauche} preference of
1,2-difluoroethane~\cite{Harris1977,Durig1992DFE}.

\begin{figure}[htbp]
\centering
\includegraphics[width=\linewidth]{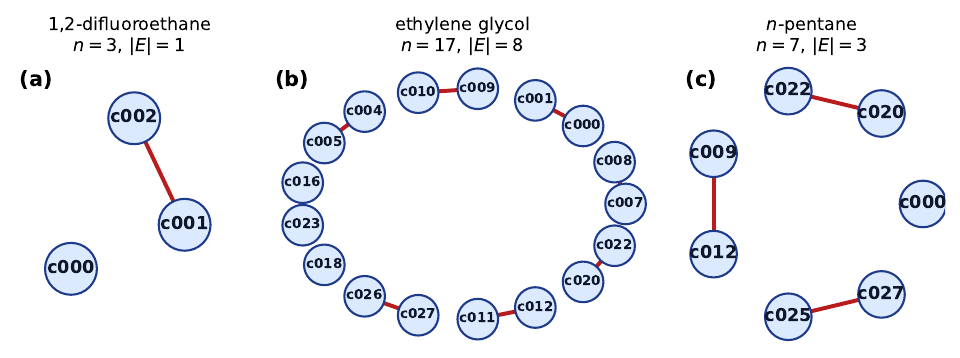}
\caption{Ambiguity-graph summary for the audited, imaginary-mode-validated
production ensembles at $\alpha=0.05$ under the full achiral stack
(IR${+}$Raman${+}$rot${+}f_{\bm\mu}$ where available). Solid red edges are
declared mirror pairs. Every retained edge is mirror-protected:
(a)~difluoroethane, (b)~ethylene glycol, and (c)~\emph{n}-pentane. Node labels
are the conformer identifiers; the conformer inventory table in the
supplementary material maps each to its torsional descriptor and relative
energy.}
\label{fig:pilot-graphs}
\end{figure}

\begin{table}[htbp]
\centering
\caption{Unresolved-pair fraction $U_S(0.05)$ and edge count for single
modalities and the complete achiral stack after quasi-RRHO
$5\,\mathrm{kcal\,mol^{-1}}$ retention. Complete intermediate modality
combinations are reported in the supplementary material.}
\label{tab:pilot-refinement}
\begin{tabular}{@{}lccc@{}}
\toprule
Modality set & 1,2-C$_2$H$_4$F$_2$ & ethylene glycol & \emph{n}-pentane\\
\midrule
IR & 0.333 (1) & 0.059 (8) & 0.143 (3)\\
Raman & 0.333 (1) & 0.059 (8) & 0.143 (3)\\
rotational $(A,B,C)$ & 0.333 (1) & 0.485 (66) & 0.333 (7)\\
dipole $f_{\bm\mu}$ & 0.333 (1) & 0.074 (10) & 0.143 (3)\\
complete achiral stack & 0.333 (1) & 0.059 (8) & 0.143 (3)\\
\bottomrule
\end{tabular}
\end{table}

\subsection{Covariance robustness and modality selection}
\label{sec:covariance-robustness}

Under the working covariance, starting from IR, every achiral candidate
(Raman, rotation, and $f_{\bm\mu}$) removed zero edges on all three molecules:
the residual edges (1/8/3 under the primary uncapped stack) are mirror-related
pairs whose separation vanishes for every reflection-invariant channel
(Table~\ref{tab:pilot-refinement}). Selection over the present candidate set
(Section~\ref{sec:design}) is therefore flat; once it is exhausted, the natural
next structural recommendation is a parity-sensitive modality (VCD, ROA, or
microwave three-wave mixing)~\cite{Nafie2011,Patterson2013}, and no
information-gain ranking among those alternatives was computed. That
recommendation concerns pure-state pairwise separability, not mixture or
kinetics recoverability (Section~\ref{sec:discussion}). This flatness is a
property of the working covariance rather than of the candidate set: under the
combined stress-test covariance the same candidates differ sharply, and
selection becomes informative.

The primary refinement table uses a working covariance with spectral FWHM
$8\wavenum$, peak SNR $50$, local frequency-shift standard deviation
$3\wavenum$, and model-error fraction $\rho_{\mathrm m}=0.05$. One-factor and
combined stress-test sweeps of that covariance are reported in the
supplementary material. The combined stress test (FWHM $16\wavenum$, SNR $20$,
$\sigma_\nu=12\wavenum$, $\rho_{\mathrm m}=0.30$) is not an independently
calibrated instrument model: it takes the pessimistic end of each one-factor
sweep at once, so it is the most conservative corner of the predefined grid. DFE and ethylene glycol remained at their mirror floors
across all one-factor and combined stress-test settings. For \emph{n}-pentane,
the combined stress-test IR-only graph had seven edges, four of which were
non-mirror ambiguities, while IR+Raman, IR+rot, IR+$f_{\bm\mu}$, and the full achiral
stack restored the three mirror-pair floor. Figure~\ref{fig:npentane-stress-graphs}
resolves those edges chemically: the four finite-noise stress-test pairs are
the cross links between the $g^\mp g^\mp$ and $g^\mp T$ mirror families
($P_e\approx0.093$ under IR alone), and each added achiral channel removed
exactly the same four pairs, while the three
mirror edges remained at $P_e=\tfrac12$. The channels differed substantially in
margin: adding rotation gave $P_e\approx5\times10^{-8}$ and $f_{\bm\mu}$ gave
$P_e<10^{-100}$, whereas the vibrational rescue is comparatively narrow: Raman
alone reached $P_e\approx2.6\times10^{-4}$ and IR+Raman
$P_e\approx1.0\times10^{-4}$. The
\emph{n}-pentane Raman activities computed for this work therefore answered a
question the earlier IR-only ensemble could not: the stressed vibrational
ambiguity is removable within vibrational spectroscopy alone, without recourse
to a rotational experiment. Pairwise Bayes errors are tabulated in
the supplementary material. This illustrates the distinction between exact
symmetry protection and finite-resolution ambiguity: added achiral measurements do not remove exact mirror
degeneracy, but they can remove finite-resolution non-mirror ambiguities
introduced by a combined stress-test covariance. Here $f_{\bm\mu}$ was treated as
a feature-ablation channel; because principal-axis dipole magnitudes presuppose
rotational resolution, the physically implementable channel is the joint
rotational package (rotation and $f_{\bm\mu}$ together), as discussed in
Section~\ref{sec:discussion}.

\begin{figure}[htbp]
\centering
\includegraphics[width=\linewidth]{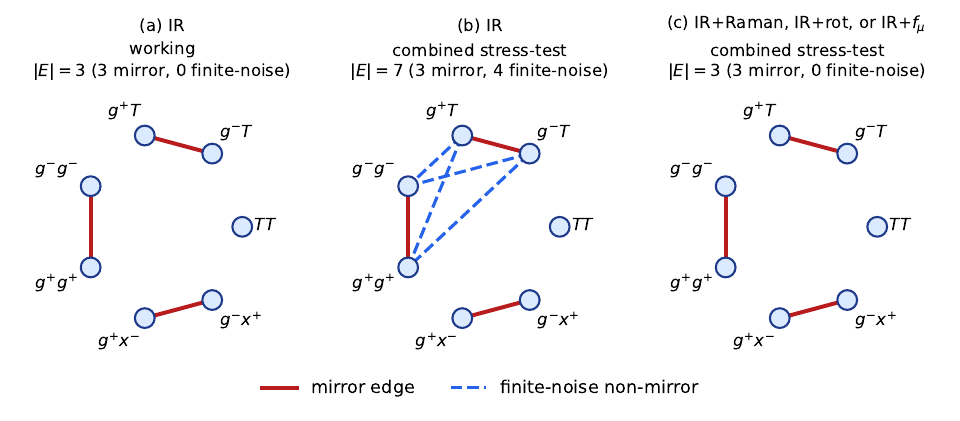}
\caption{Pair-resolved \emph{n}-pentane ambiguity graphs, drawn with identical
node positions and labeled by torsional descriptor. Solid red edges are
declared mirror pairs ($P_e=\tfrac12$); dashed blue edges are finite-noise
non-mirror ambiguities ($P_e>\alpha=0.05$). (a)~IR under the working
covariance. (b)~IR under the combined stress-test covariance. (c)~IR+Raman,
IR+rotation, and IR+$f_{\bm\mu}$ under the stress-test covariance, which share
this topology; their margins differ and are given in
Section~\ref{sec:covariance-robustness}.}
\label{fig:npentane-stress-graphs}
\end{figure}

\subsection{Spectral-window resolution of the stressed \emph{n}-pentane ambiguity}
\label{sec:window-design}

The candidate set of Section~\ref{sec:design} need not consist of whole
modalities. Because the stressed \emph{n}-pentane quotient edge is the only
non-mirror ambiguity anywhere in this study, it supports a concrete test of
window-level design: which sub-experiment (a single vibrational band,
the rotational constants, or $f_{\bm\mu}$) suffices to remove it?
Table~\ref{tab:npentane-window-design} evaluates each candidate on its own under
the combined stress-test covariance, using the five fixed bands of
Section~\ref{sec:design} for both IR and Raman. (Under the working covariance
this edge is absent, so window sufficiency is defined only against the stressed
ensemble.)

\begin{table}[htbp]
\centering
\caption{Targeted-edge rescue analysis for the stressed \emph{n}-pentane quotient edge $\mathrm{gg}$--$\mathrm{gT}$ under the combined stress-test covariance (FWHM $16\wavenum$, SNR~$20$, $\sigma_\nu=12\wavenum$, $\rho_m=0.30$). Each row is a single-candidate measurement evaluated alone against this ambiguity. $P_e(\mathrm{gg},\mathrm{gT})$ is the Bayes error of a cross link of the quotient edge under the candidate alone (the four cross links coincide); ``open'' counts how many of the four node-level edges remain at $P_e>\alpha=0.05$. The $1800$--$2800\wavenum$ windows are empty-window controls (no \emph{n}-pentane fundamentals; $P_e=\tfrac12$). The smallest squared separation over all chemically distinct pairs under each candidate is tabulated in the supplementary material.}
\label{tab:npentane-window-design}
\scriptsize
\setlength{\tabcolsep}{3pt}
\begin{tabular}{@{}lccc@{}}
\toprule
Candidate & $P_e(\mathrm{gg},\mathrm{gT})$ & open stress edges & rescues?\\
\midrule
IR $0$--$500\wavenum$ & 0.487 & 4 & no\\
IR $500$--$1000\wavenum$ & 0.254 & 4 & no\\
IR $1000$--$1800\wavenum$ & 0.260 & 4 & no\\
IR $1800$--$2800\wavenum$ (empty) & 0.500 & 4 & no\\
IR $2800$--$3800\wavenum$ & 0.171 & 4 & no\\
IR (full band) & 0.0926 & 4 & no\\
Raman $0$--$500\wavenum$ & 0.0651 & 4 & no\\
Raman $500$--$1000\wavenum$ & 0.0458 & 0 & yes\\
Raman $1000$--$1800\wavenum$ & 0.0384 & 0 & yes\\
Raman $1800$--$2800\wavenum$ (empty) & 0.500 & 4 & no\\
Raman $2800$--$3800\wavenum$ & 0.0342 & 0 & yes\\
Raman (full band) & $2.6\times10^{-4}$ & 0 & yes\\
rotational $(A,B,C)$ & $1.3\times10^{-7}$ & 0 & yes\\
dipole $f_{\bm\mu}$ & $<10^{-100}$ & 0 & yes\\
\midrule
\multicolumn{4}{@{}l@{}}{\emph{All ``yes'' candidates tie; $500$--$1000\wavenum$ Raman is the narrowest.}}\\
\bottomrule
\end{tabular}
\end{table}

Three features of that table matter for design practice. First, no IR window
resolved the edge, and neither did the full IR band
($P_e\approx0.093$): the failure is not localized in
a region that a better-chosen infrared experiment could target. Second, the
information that resolved it sat in the Raman channel. Of the stressed
pair's total additive squared Mahalanobis separation $d^2$ under the
block-diagonal stress-test covariance, the Raman channel carried $87\%$ and the
IR channel $13\%$, and three
of the five Raman windows ($500$--$1000$, $1000$--$1800$, and
$2800$--$3800\wavenum$) each removed the edge on their own, as did the full
Raman band with a far larger margin ($P_e\approx3\times10^{-4}$). The rescue is
thus not unique: the
three Raman subwindows, the full Raman band, the rotational constants, and
$f_{\bm\mu}$ each resolved all four node-level realizations of the stressed
quotient edge, and no acquisition costs or tie-breaking rule were defined. This tie is a
statement about the targeted edge, not a global minimum-$U$ ordering of the
full graph: rotational constants alone, for example, removed this edge
decisively yet left seven \emph{n}-pentane edges unresolved on their own
(Table~\ref{tab:pilot-refinement}). A global window-level
design would instead evaluate the full-graph objective $U_{S\cup\{b\}}(\alpha)$
of Eq.~\eqref{eq:minU} for each candidate against the complete edge set. Among the tested
Raman subwindows, $500$--$1000\wavenum$ was the narrowest sufficient window,
while the $2800$--$3800\wavenum$ window offered a slightly larger separation
margin. Third, the rescues differed in kind: the
single-window Raman rescues cleared the tolerance only narrowly
($P_e=0.034$--$0.046$ against $\alpha=0.05$) and would not survive a modest
tightening of $\alpha$ or a further inflation of the covariance, whereas the
full Raman band, the rotational constants, and $f_{\bm\mu}$ cleared it decisively.
Window-level selection therefore identified the narrowest tested sufficient
Raman window, but the ranking among near-threshold windows is a property of the
declared covariance and tolerance rather than a robust chemical claim. The four
stressed node edges are equivalent under the mirror symmetry of the two
families, so each candidate yields a single distinct Bayes error; the
global-minimum separations over all chemically distinct pairs are given in the
supplementary material.

The three sufficient Raman windows have a common physical origin, summarized
in three statements (per-mode values in the supplementary material,
Table~S7). First, the discrimination was carried
almost entirely by redistribution of Raman activity among skeletal and C--H
modes rather than by large frequency shifts: the largest between-family
activity ratios exceeded a factor of two while the corresponding mode frequencies
differed by at most a few $\wavenum$. Second, the discriminating motions are
backbone-coupled CH$_2$ rocks ($500$--$1000\wavenum$), CH$_2$ wag/scissors
modes whose activity is reshuffled at conserved window-integrated total
($1000$--$1800\wavenum$), and symmetry-sensitive C--H stretches whose
combinations fragment in the lower-symmetry gT geometry
($2800$--$3800\wavenum$). Third, these same modes carry small or
family-insensitive IR dipole derivatives, which is why the torsional signature
is invisible to IR in every window.

As elsewhere in Section~\ref{sec:pilot}, these are within-method separations
under the declared covariance (Section~\ref{sec:discussion}). The
$1800$--$2800\wavenum$ window contains no \emph{n}-pentane fundamentals, so
its entries are the empty-window limit $P_e=\tfrac12$ rather than a measurement
statement.

\subsection{Population identifiability and numerical conditioning}
\label{sec:population-demo}

We separate two questions the literature sometimes conflates: exact
identifiability (is the population map injective?) and numerical conditioning (how
well can the identifiable contrasts be recovered?).

\emph{Exact rank, stated analytically.}
After mirror symmetrization the achiral additive IR and Raman columns of two
mirror partners are identical. Consequently, $M$ has the null vector
$e_i-e_j$ for every declared mirror pair $(i,j)$, and each such vector lies in
the simplex tangent $\one^{\Trans}h=0$. These $n_{\mathrm{mirror}}$ independent
directions imply
$\rank([M;\one^{\Trans}])\le n_m-n_{\mathrm{mirror}}$ exactly, so separate
mirror-partner populations are exactly unidentifiable under achiral
additive spectra, while mirror-collapsed class populations are the natural
identifiable target. Numerically, the rank reached this upper bound for all
three ensembles at the stated SVD tolerance (two for DFE, nine for ethylene
glycol, and four for \emph{n}-pentane IR with $n=7$;
Table~S11), indicating no additional
class-level dependencies at that numerical resolution. The analytic argument
certifies the mirror null directions; the numerical rank assesses whether
additional dependencies occur among the collapsed classes.

\emph{Numerical conditioning of the identifiable subspace.}
Exact rank says nothing about how well the identifiable class contrasts can be
estimated. Table~\ref{tab:population-conditioning} reports the whitened
constrained Fisher information
$\Fish^{(\Delta)}=U^{\Trans}M^{\Trans}\Sig^{-1}MU$ (with $U$ an orthonormal
simplex-tangent basis) for both the full-node design and the mirror-collapsed
class design. Here $M$ and $\Sig$ were built in the same signal units: the
columns of $M$ are the unnormalized additive conformer signals (no
pure-spectrum area normalization), and $\Sig$ was re-created by the spectral
covariance algorithm of Section~\ref{sec:noise} applied to the unnormalized
population mean $\bar\mu=\tfrac1n\sum_i M_{:,i}$, with stacked modalities kept in
calibrated relative-intensity units; the reported eigenvalues and variance bounds
are therefore dimensionally consistent. The full-node design has at least
$n_{\mathrm{mirror}}$ exact zero eigenvalues (one for DFE, eight for ethylene
glycol, three for \emph{n}-pentane); numerically, no additional near-zero modes
were detected at the reported tolerance. On the
mirror-collapsed class design $\Fish^{(\Delta)}$ was positive definite and
numerically well conditioned within the declared calibrated-scale working model:
the Fisher condition number was $\kappa\le7.8$ for all three molecules under IR
(whitened-design condition number $\sqrt{\kappa}\lesssim2.8$), and the
basis-invariant worst-case variance bound over unit-norm class contrasts,
$1/\lambda_{\min}(\Fish^{(\Delta)})=\max_{\|v\|=1}v^{\Trans}(\Fish^{(\Delta)})^{-1}v$,
was at most $3\times10^{-3}$. Adding
Raman lowered ethylene glycol's condition number from $7.78$ to $6.14$ and
\emph{n}-pentane's from $5.89$ to $4.77$, and tripled \emph{n}-pentane's
smallest class-contrast eigenvalue. Adding a modality that is redundant for
distinguishability (Table~\ref{tab:pilot-refinement}) can therefore still
sharpen population recovery, since the two criteria respond to different
functionals of the same design.

\begin{table}[htbp]
\centering
\begin{threeparttable}
\caption{Whitened constrained-Fisher conditioning of the mirror-symmetrized, \emph{unnormalized} additive population design (calibrated global scale, working covariance). Entries are for the mirror-collapsed class design: least-resolved contrast $\lambda_{\min}$, Fisher condition number $\kappa=\lambda_{\max}/\lambda_{\min}$, and the basis-invariant worst-case variance bound $1/\lambda_{\min}$ over unit-norm tangent contrasts of $\Fish^{(\Delta)}=U^{\Trans}M^{\Trans}\Sig^{-1}MU$. The full-node design carries $n_{\mathrm{mirror}}$ exact zero modes, one per declared mirror pair.}
\label{tab:population-conditioning}
\footnotesize
\setlength{\tabcolsep}{3pt}
\begin{tabular}{@{}llccccc@{}}
\toprule
Molecule & mods & $n_{\mathrm{cl}}$ & full-node & class & class & class\\
 & & & $\lambda_{\min}$ & $\lambda_{\min}$ & $\kappa$ & $1/\lambda_{\min}$\\
\midrule
1,2-difluoroethane & IR & 2 & $0$\tnote{a} & $7613$ & $1.00$ & $1.3\times10^{-4}$\\
 & IR+Raman & 2 & $0$\tnote{a} & $8792$ & $1.00$ & $1.1\times10^{-4}$\\
ethylene glycol & IR & 9 & $0$\tnote{a} & $2.2\times10^{4}$ & $7.78$ & $4.6\times10^{-5}$\\
 & IR+Raman & 9 & $0$\tnote{a} & $3.2\times10^{4}$ & $6.14$ & $3.1\times10^{-5}$\\
\emph{n}-pentane & IR & 4 & $0$\tnote{a} & $334$ & $5.89$ & $3.0\times10^{-3}$\\
 & IR+Raman & 4 & $0$\tnote{a} & $1004$ & $4.77$ & $10.0\times10^{-4}$\\
\bottomrule
\end{tabular}
\begin{tablenotes}
\footnotesize
\item[a] Exact zero by construction: identical mirror columns make $\Fish^{(\Delta)}$ singular; the residual is at the level of floating-point round-off ($<10^{-6}$ relative to $\lambda_{\max}$).
\item The last column is the worst-case variance bound over unit-norm class contrasts, $1/\lambda_{\min}(\Fish^{(\Delta)})$; unlike a diagonal of $(\Fish^{(\Delta)})^{-1}$ it is invariant to the choice of tangent basis $U$.
\end{tablenotes}
\end{threeparttable}
\end{table}

These are calibrated-scale results (known global intensity). Under the
unknown-scale model~\eqref{eq:mixture-scale} the governing criterion is instead
the stronger rank condition $\rank(M_S)=n_m$; the scale is handled by introducing
and profiling a global scale nuisance parameter, whose Schur-complement
information~\eqref{eq:profiled-scale-fisher} subtracts the direction along $\wvec$
from $\Fish^{(\Delta)}$ (an information-projection, not a set-theoretic quotient),
and the corresponding full-node rank outcome is unchanged by the exact mirror
degeneracies. On the mirror-collapsed designs, however, $M$ was full column rank
for every molecule and modality set. With the same working covariance, at equal
class populations and unit global scale, the profiled-scale $\lambda_{\min}$
remained positive for all six designs, so unknown scale introduced no additional
local class-population null direction (values in Table~S11).

\subsection{Experimental consistency checks}
\label{sec:expcompare}

Three limited comparisons examined the energy ordering and selected observables
against literature values. For 1,2-difluoroethane, the production calculation
recovered the experimentally established \emph{gauche}
preference~\cite{Harris1977}: $\Delta H_{\mathrm{calc}}
=1.07\,\mathrm{kcal\,mol^{-1}}$ and $\Delta G_{298,\mathrm{calc}}
=0.88\,\mathrm{kcal\,mol^{-1}}$, compared with the experimental Raman value
$\Delta H_{\mathrm{exp}}=1.98\pm0.08\,\mathrm{kcal\,mol^{-1}}$. The ordering is
correct, although the calculated enthalpy difference is smaller by
$0.91\,\mathrm{kcal\,mol^{-1}}$.

For the assigned 1,2-difluoroethane~\cite{Craig1997,Takeo1986,CCCBDB,Butcher1971}
and ethylene-glycol~\cite{ChristenMuller2003,Christen2001,MullerChristen2004,CDMS2005,CDMS}
conformers, production rotational constants agreed with experiment to
approximately $1\%$ across the twelve tabulated constants. Dipole magnitudes
agreed within approximately $2\%$ for 1,2-difluoroethane and $7\%$ for ethylene
glycol, although individual ethylene-glycol principal-axis components showed
larger errors. Passing the experimental mean vectors through the same decision
rule reproduced the qualitative computational result: the chemically distinct
ethylene-glycol $aGg'$/$gGg'$ pair was unresolved by rotational constants alone
but resolved by the dipole-component feature, whereas the difluoroethane
\emph{anti}/\emph{gauche} pair was resolved by either. Full numerical
comparisons and parity plots are given in the supplementary material.

\section{Discussion and limitations}
\label{sec:discussion}

All reported $U_S$, edge sets, $P_e$ values, and design nominations are
sensitivity results under the curated production ensembles and the explicit
working covariance (Table~S1). They are not empirical
out-of-method error rates or a chemical-space survey. The assumptions and
scope limits that qualify those results are collected here.

\paragraph{Observation model.}
The closed-form Bayes error~\eqref{eq:bayeserror} and Fisher
information~\eqref{eq:fisher} assume shared-covariance Gaussian observations;
heavy-tailed or multiplicative noise would require Chernoff information and a
numerical Bayes error. The error and edge monotonicity statements,
Eqs.~\eqref{eq:errmonotone} and~\eqref{eq:edge-monotonicity}, are
data-processing consequences and survive that change. The refinement inclusion in
Theorem~\ref{thm:refinement} and the Bayes-error/edge monotonicity relations do
not require conditional independence: adding any observation, dependent or
not, can only refine the exact quotient and remove ambiguity edges. Conditional independence
is required only for equality of $\sim_{S\cup\{b\}}$ with the intersection
$\sim_S\cap\sim_b$, for additive squared Mahalanobis separation, and for the
additive Fisher-information relation~\eqref{eq:fisheradd}. Shared theoretical
bias violates that factorization, which is why $\Sig_S$ may carry cross-modality
blocks. The present numerical stack is block diagonal across modalities and
therefore does not estimate shared cross-channel DFT bias. Mixture analysis assumes linear intensity response in
population (dilute Beer--Lambert / standard Raman) and does not license
$(A,B,C)$ or $f_{\bm\mu}$ as mixture columns.

\paragraph{Working covariance and inverse crime.}
Graphs are built from B3LYP means with the fixed covariance algorithm of
Section~\ref{sec:noise}. Recomputing a stratified subset of two conformers per
molecule at PBE0-D3(BJ)/def2-TZVP gave mode-frequency discrepancies of
approximately $11$--$13\wavenum$ relative to B3LYP, concentrated in the
fingerprint region, and cross-method intensity-error fractions of
$0.06$--$0.10$; these values define the conservative cross-method covariance
used below and are tabulated in the supplementary material (Table~S8). The
two-conformer PBE0 discrepancy table and the
six-case cross-method pseudo-observation diagnostic
(Table~\ref{tab:reassign-summary}) diagnose parts of that model; they are not
substituted into $\Sig_S$ and are not a validation accuracy.
The experimental energetics, rotational/dipole comparisons, and experimental-mean
ambiguity check in Section~\ref{sec:expcompare} are likewise consistency checks on
pieces of the forward model, not calibration of the full multi-conformer graph.

\paragraph{Nuisance parameters and alignment.}
The reported graphs fix the measurement nuisance parameters $\eta_a$ of
Eq.~\eqref{eq:obsmodel}: a fixed global frequency scale $s_{\mathrm{harm}}=0.975$,
no profiled or marginalized shift, and area normalization. To quantify the effect
of an unknown calibration we recomputed, for every template pair, the worst-case
bounded-nuisance separation: the minimum Mahalanobis distance between the two
transformed template manifolds when each candidate may choose its own frequency
scale and rigid shift within the declared calibration window (the explicit
definition, optimization grid, and the distinct decision problems that
calibration can define are given in the supplementary material). Since the fixed calibration is admissible, profiling can
only lower the separation and hence only add edges relative to the fixed graph.
The recomputation was repeated under a conservative cross-method covariance, with
spectral variance terms raised componentwise to the larger of the working value
and the two-conformer PBE0--B3LYP estimates of
Table~S8
($\sigma_\nu\approx11$--$13\wavenum$, $\rho_m\approx0.06$--$0.10$; construction
in the supplementary material). In every combination of
modality set (IR or full stack), covariance (working or cross-method), and
profiling mode (fixed, scale, or scale$+$shift), no non-mirror edges appeared, and
the smallest chemically distinct separation remained far above the resolvability
threshold $d^\ast\approx3.29\,\sigma$ (worst case $d_{\min}\approx19\,\sigma$,
\emph{n}-pentane IR under the cross-method covariance with scale profiling; the
full grid is tabulated in the supplementary material). Within both the declared
working model and this conservative cross-method covariance, the separation of
chemically distinct templates is therefore robust to calibration
uncertainty; the residual ambiguity is the symmetry floor.

Applying the identical treatments directly to the six cross-method
pseudo-observation queries (each PBE0 spectrum a fixed query against the full
B3LYP library, with every candidate template given its own calibration before
scoring) gives Table~\ref{tab:reassign-summary}. Under fixed calibration only one
of six queries recovered the intended
representative, with the errors falling in chemically different classes; allowing
per-candidate scale$+$shift alignment recovered three of six under the working
covariance (four of six counting assignment to the correct achiral class, one
additional error being a mirror-partner interchange) and five of six under the
conservative cross-method covariance, with the correct achiral class recovered in
all six. These
results are consistent with a substantial calibration-related component in the
cross-method assignment failures that calibration profiling largely
removes; the single residual error under the strongest treatment was a
mirror-partner interchange within one achiral class (the symmetry floor). This is an illustrative diagnostic on a
two-conformer-per-molecule cross-method subset, not a validation accuracy; a
production covariance fit on a larger cross-method set with displacement-vector
mode matching remains the outstanding extension.

\begin{table}[t]
\centering
\caption{Cross-method reassignment summary for the six PBE0$\to$B3LYP
pseudo-observation queries: number recovering the intended representative and
number recovering the correct achiral class under three covariance/calibration
treatments. The full query-by-query reassignment table, including the scale-only
treatments and the fixed cross-method baseline, is given in the supplementary
material.}
\label{tab:reassign-summary}
\footnotesize
\begin{tabular}{@{}lcc@{}}
\toprule
Treatment & correct representative & correct achiral class\\
\midrule
fixed, working covariance & 1/6 & 1/6\\
scale$+$shift, working covariance & 3/6 & 4/6\\
scale$+$shift, cross-method covariance & 5/6 & 6/6\\
\bottomrule
\end{tabular}
\end{table}

\paragraph{Rotational constants and the dipole feature.}
Because principal-axis dipole components are obtained from Stark measurements on
resolved rotational transitions, $f_{\bm\mu}$ presupposes a rotational
experiment; the IR$+f_{\bm\mu}$ result is therefore a mathematical
feature-ablation test, and the physically implementable object is a joint
rotational-spectroscopy package or a conditional sequence in which $f_{\bm\mu}$
becomes available only after rotational resolution (detailed in the
supplementary material).

\paragraph{Electronic structure and thermochemistry.}
Quasi-RRHO reduces the most direct low-frequency entropy artifact but does not
replace coupled multidimensional hindered-rotor thermochemistry. The working
spectral scale $s_{\mathrm{harm}}=0.975$ is a declared working harmonic-frequency
factor used for the measurement operator; it is not specifically parameterized
for B3LYP-D3(BJ)/def2-TZVP in NWChem
(cf.\ Merrick et al.~\cite{MerrickMoranRadom2007}). \emph{n}-pentane Raman
activities are now computed at the study level, so the stress-test question is
answered directly (Sections~\ref{sec:covariance-robustness}
and~\ref{sec:window-design}); two limitations remain. The PBE0 cross-method
subset contains no Raman, so the cross-method discrepancy estimate, and hence
the conservative cross-method covariance, is calibrated
on IR residuals alone and is applied to the Raman channel by assumption. And
Raman intensities are static-limit polarizability derivatives converted to a
single excitation wavelength, so resonance and frequency-dispersion effects are
outside the model.

\paragraph{Conformer search and curation.}
Five CREST searches per molecule have incomplete overlap. Promoting every unique
screened representative and excluding imaginary-mode non-minima is necessary
before graph summaries are interpreted, but it cannot prove that every thermally
accessible minimum was found. A leave-one-search-out rematch of archived CREST
geometries to production representatives (supplementary material) retained
mirror-only residual edges in every omit-one case examined here, while
\emph{n}-pentane still accumulated unique matches across the five-search
archive. That rematch uses a deliberately permissive $0.85\,\angstrom$ Kabsch
threshold, far coarser than the $0.125\,\angstrom$ de-duplication radius, so
it measures coverage rather than one-to-one assignment, and a single
ethylene-glycol search can therefore appear to cover many production
representatives despite the lower per-search merged-set overlap reported above.
A permutation- and torsion-aware matcher with reported threshold sensitivity
would be required to convert that coverage statement into a per-conformer
recovery claim. Larger validation sets could draw on GEOM~\cite{GEOM2022} or
QVib~\cite{QVib2026}, subject to their release formats and licensing.

\paragraph{Design nominations.}
Parity-sensitive nominations (VCD, ROA, microwave three-wave mixing) follow from
mirror protection under achiral maps. They are not computed information-gain
comparisons among those modalities, and they do not by themselves establish
kinetic isolation or equilibrium-mixture recoverability for interconverting or
enantiomerically canceling ensembles.

\section{Conclusion}
\label{sec:conclusion}

Exact observational equivalence and finite-noise ambiguity answer different
questions: the former defines a quotient, whereas the latter defines a
non-transitive graph whose edges can only be removed by additional data.
For additive, unnormalized spectra, population recovery has corresponding rank
and constrained-Fisher criteria, with a stricter condition when the global
intensity scale is unknown.

All numerical results in the three-molecule calculation are conditional on the
declared covariance and calibration models. Under the within-method working
covariance, the calculated IR means separated all non-mirror pairs and the
residual full-stack edges were the mirror pairs fixed by the symmetry
construction; separate mirror-partner populations are exactly unidentifiable
under achiral additive spectra, while the mirror-collapsed class populations
were numerically well conditioned
(Table~\ref{tab:population-conditioning}). The accompanying cross-method
diagnostic showed that fixed calibration failed on most queries while
scale-and-shift profiling under a conservative cross-method covariance recovered
the correct achiral class in every case (Table~\ref{tab:reassign-summary}),
consistent with a substantial calibration-related component in the cross-method
assignment error, and worst-case joint spectral nuisance profiling under both
covariances added no non-mirror edges. On the exact
observational quotient the working-covariance graphs had no edges, so achiral
measurement design is trivial there; it becomes nontrivial only under the
combined stress-test covariance for \emph{n}-pentane, which created four
IR-only non-mirror quotient edges (the $g^\mp g^\mp$--$g^\mp T$ cross links of
Figure~\ref{fig:npentane-stress-graphs}). Raman, rotation, and $f_{\bm\mu}$ each
removed them, and the per-window analysis localized the discriminating
information: no infrared band sufficed, whereas three Raman bands did
individually, with the $500$--$1000\wavenum$ Raman window the
narrowest tested sufficient window
(Table~\ref{tab:npentane-window-design}). A modality that is redundant for
distinguishability under the working covariance is thus not redundant under a
stressed one. Parity-sensitive measurements remain the natural candidates for the
residual mirror edges. The principal validation limitation is the illustrative
two-conformer-per-molecule cross-method subset: a production
covariance fit on a larger cross-method set with displacement-vector mode
matching, a cross-method subset that includes Raman, and computed
information-gain comparisons among parity-sensitive modalities are the
priorities for extending the numerical demonstration.

\section*{Supplementary Material}
See the supplementary material for the complete computational settings and
covariance-construction algorithm; conformer-search, ensemble, and covariance
robustness analyses; pair-resolved \emph{n}-pentane stress-test results and
Raman-mode data; cross-method discrepancy, calibration, and reassignment
details; population-rank calculations and derivations; experimental comparison
tables; and the expanded conformer inventory. The complete reproducibility
archive, including the processed observables, configurations, analysis code,
software-environment records, integrity records, and archived raw
conformer-search and electronic-structure files, is publicly available in the
Zenodo deposit identified in the Data Availability statement.

\section*{Conflict of Interest}
The authors have no conflicts to disclose.

\section*{Author Contributions}
Megan Simons: Conceptualization, Methodology, Software, Validation, Formal
analysis, Investigation, Data curation, Writing -- original draft, Writing --
review \& editing, Visualization, Project administration. Jonathan Washburn:
Conceptualization, Methodology, Software, Investigation, Writing -- review \&
editing, Supervision.

\section*{Data Availability}
The data and analysis code supporting the findings of this study are openly
available in the Zenodo repository~\cite{SimonsWashburn2026Data}. The deposit
contains the processed observables, configurations, software-environment
records, analysis scripts used to generate the reported tables and figures, a
script-to-artifact manifest, SHA-256 integrity records, and the archived raw
conformer-search and electronic-structure inputs and outputs from which the
processed observables were derived. Every reported table and figure can
therefore be regenerated end-to-end without contacting the authors.

The open-source \texttt{conformer\_spec} implementation is released under the
MIT License, and the deposited data and documentation are released under
CC BY 4.0. Script and artifact filenames carry frozen campaign tags
(\texttt{v15}--\texttt{v23}); these are stable archive identifiers rather than
manuscript version numbers. Re-running CREST statistically replicates, rather
than exactly reproduces, the stochastic conformer search; the archived search
records and deterministic downstream analyses provide the reproducible record.

\end{document}